\documentclass[11pt]{article}

\usepackage[utf8]{inputenc}
\usepackage{amsmath, amsthm, amssymb, amsfonts}
\usepackage{cleveref, bm, commath}
\usepackage{bbm}
\usepackage{natbib}
\usepackage{graphicx}
\usepackage{xcolor}
\usepackage{multirow}
\usepackage{caption}
\usepackage{subfigure}
\usepackage{mwe}

\usepackage{enumitem}

\theoremstyle{definition}
\newtheorem{definition}{Definition}

\newtheorem{proposition}{Proposition}
\newcommand{\indep}{\rotatebox[origin=c]{90}{$\models$}}

\usepackage{tikz}
\usepackage{amsmath, amsfonts, amsthm}

\usepackage{booktabs,siunitx}

\usepackage{listings}
\usepackage[margin=1in]{geometry}

\usepackage{sectsty}
\usepackage[onehalfspacing]{setspace}
\usepackage{float}

\begin{document}

\sectionfont{\bfseries\large\sffamily}%

\subsectionfont{\bfseries\sffamily\normalsize}%

\noindent
{\sffamily\bfseries\LARGE
  Selecting and Ranking Individualized Treatment Rules with Unmeasured
  Confounding}%

\noindent
\textsf{Bo Zhang, Jordan Weiss, Dylan S.\ Small, Qingyuan Zhao*}

\noindent
\textsf{University of Pennsylvania and University of Cambridge}

\noindent
\textsf{}%

\noindent
\textsf{{\bf Abstract}: It is common to compare individualized treatment rules based on the value function, which is the expected potential outcome under the treatment rule. Although the value function is not point-identified when there is unmeasured confounding, it still defines a partial order among the treatment rules under Rosenbaum's sensitivity analysis model. We first consider how to compare two treatment rules with unmeasured confounding in the single-decision setting and then use this pairwise test to rank multiple treatment rules. We consider how to, among many treatment rules, select the best rules and select the rules that are better than a control rule. The proposed methods are illustrated using two real examples, one about the benefit of malaria prevention programs to different age groups and another about the effect of late retirement on senior health in different gender and occupation groups.}%


\vfill \noindent
\textsf{*Correspondence to: Qingyuan Zhao, Statistical Laboratory,
  University of Cambridge, Wilberforce Road, Cambridge, CB3 0WB,
  UK. Email: qyzhao@statslab.cam.ac.uk.}

\newpage

\section{Introduction}

A central statistical problem in precision medicine and health policy
is to learn treatment rules that are tailored to the
patient's characteristics. There is now an exploding literature on
individualized policy discovery; see \citet{Precision_medicine} for an
up-to-date review. Although randomized experiments remain the gold
standard for causal inference, there has been a growing interest in using
observational data to draw causal conclusions and discover
individualized treatment rules due to the increasing availability of
electronic health records and other observational data sources
\citep{moodie2012q,athey2017efficient,kallus2017recursive,M_learning,zhao2019efficient}.

A common way to formulate the problem of individualized policy
discovery is via the \emph{value
  function}, which is the expected potential outcome under a treatment rule
or regime. The optimal treatment rule is usually defined as the one
that maximizes the value function. In the single-decision setting, the
value function can be easily
identified when the data come from a randomized experiment (as long as
the probability of receiving treatment is never $0$ or $1$). When the
data come from an observational study, the value function can still be identified under the assumption that all confounders are measured. This assumption
can be further extended to the multiple-decision setting
\citep{murphy2003optimal,robins2004optimal}. In this paper we
will focus our discussion on the single-decision setting but consider
the possibility of unmeasured confounding.

With few exceptions, the vast majority
  of existing methods for treatment rule discovery from observational
  data are based on the no unmeasured confounding
  assumption. Typically, these methods first estimate the value
  function assuming no unmeasured confounding
  and then select the treatment rule that maximizes the
  estimated value function. However, it is common that a substantial
fraction of the
population appear to behave similarly under treatment or control. From
a statistical perspective and if there is truly no unmeasured
confounder, we should still attempt to estimate the treatment effect
for individuals in this subpopulation and optimize the treatment rule
accordingly. However, the optimal treatment decisions for
these individuals are, intuitively, also the most sensitive to unmeasured
confounding. It may only take a small amount of unmeasured confounding
to change the sign of the estimated treatment effects for these
individuals. From a policy perspective (especially when there is a
cost constraint), learning the ``optimal'' treatment decision for
these individuals from observational data seems likely
  to be error-prone.

\subsection{Sensitivity analysis for individualized treatment rules}
\label{sec:sens-analys-indiv}

There is a long literature on studying the sensitivity of
observational studies to unmeasured confounding, dating from
\citet{Cornfield1959}. In short, such sensitivity analysis asks
how much unmeasured confounding is needed to alter the causal
conclusion of an observational study qualitatively. In this paper, we
will study the sensitivity of individualized treatment rules to
unmeasured confounding using a prominent model proposed by
\citet{Rosenbaum1987}, where the odds ratio of receiving the treatment
for any two individuals with the same observed covariates is bounded
between $1/\Gamma$ and $\Gamma$ ($\Gamma \ge 1$; $\Gamma = 1$
corresponds to no unmeasured confounding). More specifically,
we will consider selecting and ranking individualized treatment rules
under Rosenbaum's model for unmeasured confounding.

Our investigation is motivated by the impact of effect
modification on the power of Rosenbaum's sensitivity analysis that is studied by
\citet{Hsu2013effectmodification}. A phenomenon found by
\citet{Hsu2013effectmodification} is that subgroups with larger
treatment effect may have larger design sensitivity. For example,
suppose a subgroup A has larger treatment effect than a subgroup B
based on observational data. Then, there may exist a $\Gamma > 1$ such
that, when the sample size of both subgroups go to infinity, the
probability of rejecting Fisher's sharp null hypothesis under the
$\Gamma$-sensitivity model goes to $1$ for subgroup A and $0$ for
subgroup B. Therefore, to obtain causal conclusions that are most
robust to unmeasured confounding, it may be more desirable to use a smaller
subgroup with larger treatment effect than to use a larger subgroup
with smaller treatment effect.

When comparing individualized treatment rules to a baseline, the above
phenomenon suggests that a treatment rule with smaller value may be
less sensitive to unmeasured confounding than a treatment rule with
larger value. In other words, when there is unmeasured confounding,
the ``optimal'' treatment rule might not be the one that maximizes the
value function assuming no measured confounding; in fact, there are usually
many ``optimal'' treatment rules. This is because the value function
in this case only defines a \emph{partial order} on the set of
individualized treatment rules, so two rules with different value
function assuming no unmeasured confounding may become
indistinguishable under the $\Gamma$-sensitivity model when $\Gamma >
1$. Fundamentally, the reason is that the value function is only
partially identified in Rosenbaum's $\Gamma$-sensitivity model.

As an example, let's use $r_2 \succ_{\Gamma} r_1$ (abbreviated as
$r_2 \succ r_1$ if the value of $\Gamma$ is clear from the context) to denote that
the value of rule $r_2$ is \emph{always} greater than the value of
$r_1$ when the unmeasured confounding satisfies the
$\Gamma$-sensitivity model. Then, it is possible that
\begin{itemize}
\item Under $\Gamma = 1$, $r_2 \succ r_1 \succ r_0$ (so $r_2 \succ r_0$);
\item Under some $\Gamma > 1$, $r_1 \succ r_0$ but $r_2 \not \succ r_0$.
\end{itemize}
This phenomenon occurs frequently in real data examples, see
\Cref{fig: hasse malaria hybrid} in \Cref{subsec: ranking}. Note that
the relation $\succ$ is
defined using the value function computed using the population instead
of a particular sample.

Because the value function only defines a partial order on the
treatment rules, it is no longer well-defined to estimate \emph{the}
optimal treatment rule when there is unmeasured confounding. Instead,
we aim to recover the partial ordering of a set of treatment rules or
select a subset of rules that satisfy certain statistical
properties. This problem is related to the problem of selecting and ranking
subpopulations (as a post hoc analysis for randomized experiments)
which has been extensively studied in statistics
\citep{gupta1979multiple,gibbons1999selecting}. Unfortunately, in problems
considered by the existing literature, the subpopulations always have
a \emph{total order}. For example, a prototypical problem in that
literature is to select a subset that contains the largest $\mu_i$
based on independent observations $Y_i \sim \mathrm{N}(\mu_i,1)$. It is
evident that the methods developed there cannot be directly applied to the
problem of comparing treatment rules which only bears a partial
order. Nevertheless, we will borrow some definitions in
that literature to define the goal of selecting and ranking
individualized treatment rules.

  \subsection{Related work and our approach}
  Existing methods for individualized policy discovery from observational data
  often take an \emph{exploratory} stance. They often aim to select the
  individualized treatment rule, often within an infinite-dimension
  function class, that maximized the estimated value function using
  outcome regression-based \citep{robins2004optimal, qian_murphy2011},
  inverse-probability weighting \citep{Zhao_OWL, kallus2017recursive},
  or doubly-robust estimation\citep{dudik2014doubly,
    athey2017efficient}. In order to estimate the value function, some
  parametric or semiparametric models are specified to model the outcome
  and/or the treatment selection process. To identify the value function,
  the vast majority of these approaches make the no unmeasured confounding
  assumption which may be unrealistic in many applications. The only
  exception to our knowledge is \citet{kallus2018confounding}, in which
  the authors propose to maximize the minimum
  value of an individualized treatment rule when the unmeasured
  confounding satisfies a marginal sensitivity model
  \citep{tan2006distributional,Zhao2017}. This is further extended to
  the estimation of conditional average treatment effect with unmeasured
  confounding in \citet{kallus2018interval}. Another related work is
  \citet{yadlowsky2018bounds} who consider semiparametric inference
  for the average treatment effect in Rosenbaum's sensitivity model.

  In this paper we take a different perspective. Our
  approach is based on a statistical test to compare the
  value of two individualized treatment rules when there is limited
  unmeasured confounding. Briefly speaking, we first match the treated and
  control observations by the observed covariates and then propose to
  use Rosenbaum's sensitivity model to quantify the magnitude of
  unmeasured confounding after matching (the deviation of the
  matched observational study from a pairwise randomized
  experiment). At the core of our proposal is a randomization
  test introduced by \cite{fogarty2016studentized} to compare the value
  of two individualized treatment rules in Rosenbaum' sensitivity
  model. Based on this test, we
  introduce a framework to rank and select treatment rules within a
  given finite collection and show that different statistical errors
  can be controlled with the appropriate multiple hypothesis testing
  methods.

  In principle,  our framework can be used with an arbitrary (finite)
  number of pre-specified treatment rules. In practice, it is more
  suitable for small-scale policy discovery with relatively few
  decision variables, where it is not needed to use machine
  learning methods to discover complex patterns or such methods have
  already been employed in a preliminary study to suggest a few
  candidat rules. The design-based nature of our approach makes it
  particularly useful
  for \emph{confirmatory} analyses, the importance of which is widely
  acknowledged in the policy discovery literature
  \citep[e.g.][]{kallus2017recursive,zhang2018interpretable,Precision_medicine}. Methods
  proposed in this paper thus complement the existing literature on
  individualized treatment rules by
  providing a way to confirm the effectiveness of a treatment rule
  learned from observational data and assess its robustness to
  unmeasured confounding. When there are several competing treatment
  rules, our framework further facilitates the decision maker
  to select or rank the treatment rules using flexible criteria.

The rest of the paper is organized as follows. In \Cref{sec:2} we
introduce a real data example that will be used to illustrate the
proposed methods. We then introduce some notations and discuss how to
compare two treatment rules when there is unmeasured confounding. In
\Cref{sec:multiple} we consider three questions about ranking and
selecting among multiple treatment rules. We compare our proposal with
some baseline procedures in a simulation study
\Cref{sec:simulations} and apply our method to another application
using data from the Health and Retirement Study. Finally, we conclude
our paper with some brief discussion in \Cref{sec:discussion}.

\section{Comparing treatment rules with unmeasured confounding}
\label{sec:2}

\subsection{Running example: Malaria in West Africa}
\label{sec:running-example}

The Garki Project, conducted by the World Health Organization and the
Government of Nigeria from 1969-1976, was an observational study that
compared several strategies to control
malaria. \citet{Hsu2013effectmodification} studied the effect
modification for one of the malaria control
strategies, namely spraying with an insecticide, propoxur, together
with mass administration of a drug sulfalene-pyrimethamine at high
frequency. The outcome is the difference between the frequency
of Plasmodium falciparum in blood samples, that is, the frequency of a
protozoan parasite that causes malaria, measured before and after the
treatment. Using 1560 pairs of treated and control individuals matched
by their age and gender, \citet{Hsu2013effectmodification} found that
the treatment was much more beneficial for young children
than for other individuals, if there is not unmeasured
confounding.

More interestingly, they found that, despite the reduced
sample size, the 447 pairs of young children exhibit a treatment
effect that is far less sensitive to unmeasured confounding bias than
the full sample of 1560 pairs. So from a policy perspective, it may be
preferable to implement the treatment only for young children rather
than the whole population. In the rest of this paper we will
generalize this idea to selecting and ranking treatment rules. We will
use the matched dataset in \citet{Hsu2013effectmodification} to
illustrate the definitions and methodologies in the paper; see the
original article for more information about the Garki Project dataset
and the matched design. A different application concerning the effect
of late retirement on health outcomes will be presented in
\Cref{sec:application} near the end of this article.

\subsection{Some notations and definitions}
\label{sec:some-notat-defin}

We first introduce some notations in order to compare treatment rules
when there is unmeasured confounding. Let ${X}$
denote all the pre-treatment covariates measured by the
investigator. In the single-decision setting considered in this paper,
an individualized treatment rule (or treatment regime) $d$ maps a
vector of pre-treatment covariates $ X$ to the binary treatment decisions,
$\{0, 1\}$ ($0$ indicates control and $1$ indicates treatment). In our
running example, we shall consider six treatment rules,
$r_0,r_1,\cdots,r_5$, where $r_i$ assigns treatment to the youngest $i
\times 20\%$ of the individuals. Specifically, the minimum, $20\%,
40\%, 60\%, 80\%$ quantiles, and maximum of age are $0 \,
(\text{newborn})$, $7$, $20$, $31$, $41$, and $73$ years old.

Let
$Y$ be the outcome and $Y(0),Y(1)$ be the potential outcomes under
control and treatment. The potential outcome under a treatment rule
$d$ is defined, naturally, as $Y(d) = Y(0) 1_{\{d(X)=0\}} + Y(1)
1_{\{d(X)=1\}}$. A common way to compare treatment rule is to use its
value function, defined as the expected potential outcome under that
rule, $V(d) = \mathbb{E}[Y(d)]$. The \emph{value difference} of two treatment
rules, $r_1$ and $r_2$, is thus
\begin{equation} \label{eq:value-diff}
  \begin{split}
    V(r_2) - V(r_1) &= \mathbb{E}[Y(r_2) - Y(r_1) \,|\, r_2 \neq
    r_1] \cdot \mathbb{P}(r_2 \neq r_1) \\
    &= \mathbb{E}[Y(1) - Y(0) \,|\, r_2 > r_1] \cdot \mathbb{P}(r_2 > r_1)
    - \mathbb{E}[Y(1) - Y(0) \,|\, r_2 < r_1] \cdot \mathbb{P}(r_2 < r_1),
  \end{split}
\end{equation}
where for simplicity the event $r_1( X) \ne r_2( X)$ is
abbreviated as $r_1 \ne r_2$ (similarly for $r_1 < r_2$ and $r_1 >
r_2$). Note that the event $r_2 > r_1$ is the same as $r_2 = 1,r_1 =
0$ because the treatment decision is binary. One of the terms on the right
hand side of \eqref{eq:value-diff} will become $0$ if the treatment
rules are nested. In the malaria example, $r_0 \le r_1 \le \cdots \le
r_5$, so the value difference of the rules $r_1$ and $r_2$  can be written as
\[
  V(r_2) - V(r_1) = \mathbb{E}[Y(1)-Y(0)\, |\, \text{Age} \in [7, 20)] \cdot
  \mathbb{P}(\text{Age} \in [7, 20)).
\]
In this case, testing the sign of $V(r_2) - V(r_1)$ is equivalent to
testing the sign of the conditional average treatment effect,
$\mathbb{E}[Y(1)- Y(0)\,|\,r_2 > r_1]$.

The definition of the value function depends on the potential
outcomes. To identify the value function using observational data, it
is standard to make the following assumptions \citep{Precision_medicine}:
\begin{enumerate}
\item Positivity: $\mathbb{P}(A=a \,|\, X =  x) > 0$ for all $a$
  and $ x$;
\item Consistency (SUTVA): $Y = Y(A)$;
\item Ignorability (no unmeasured confounding): $Y(a)\, \indep \, A \, | \,  X$
  for all $a$.
\end{enumerate}
Under these conditions, it is straightforward to show that the value
function is identified by \citep{qian_murphy2011}
\[
  V(d) = \mathbb{E} \bigg[ \frac{Y I(A=d(X))}{\pi(A,X)} \bigg],
\]
where $I$ is the indicator function of an event and $\pi(a,x) =
\mathbb{P}(A=a|X=x)$ is the propensity score.


The value function gives a natural and total order to the
treatment rules. If the above identification assumptions hold, the
value functions can be identified and thus this order can be
consistently estimated as the sample size increases to infinity. In
general, it is impossible to recover this order when there is
unmeasured confounding. However, if the magnitude of unmeasured
confounding is bounded according to a sensitivity model (a collection
of distributions of the observed variables and unobserved potential
outcomes), it is possible to partially identify difference between the
value of two treatment rules and thus obtain a partial order.



\begin{definition}
  \label{def: prec_gamma}
  Let $r_1$ and $r_2$ be two treatment rules that map a vector of
  pre-treatment covariates $X$ to a binary treatment decision $\{0,
  1\}$, and $V(r_1)$, $V(r_2)$ their corresponding value
  functions. Given a sensitivity analysis model indexed by $\Gamma$, we say
  that the rule $r_1$ is dominated by $r_2$ with a margin $\delta$ if
  $V(r_2) - V(r_1) > \delta$ for all distributions in the sensitivity
  analysis model. We denote this relation as $r_1 \prec_{\Gamma,\delta} r_2$
  and furthere abbreviate it as $r_1 \prec_\Gamma r_2$ if $\delta =
  0$. We denote $r_1 \not \prec_{\Gamma} r_2$ if $r_1$ is not
  dominated by $r_2$ with margin $\delta = 0$.
\end{definition}

Notice that the partial order should be defined in terms of the
partially identified interval for $V(r_2) - V(r_1)$ instead of the
partially identified intervals for $V(r_1)$ and $V(r_2)$. This is
because the same distribution of the unobserved potential outcomes
needs to be used when computing the partially identified interval for
$V(r_2) - V(r_1)$, so it is not simply the difference between the
partially identified intervals for the individual values (the easiest
way to see this is to take $r_1 = r_2$). We thank an anonymous
reviewer for pointing this out.


It is easy to see that $\prec_{\Gamma}$ is a strict partial order on
the set of treatment rules because it satisfies irreflexivity (not $r_1
\prec_{\Gamma} r_1$),
transitivity ($r_1 \prec_{\Gamma} r_2$ and $r_2 \prec_{\Gamma} r_3$
imply $r_1 \prec_{\Gamma} r_3$), and asymmetry ($r_1 \prec_{\Gamma}
r_2$ implies not $r_2 \prec_{\Gamma} r_1$). In Rosenbaum's sensitivity
model be introduced in the section below, $\Gamma = 1$ corresponds
to no unmeasured confounding and thus the relationship
$\prec_{\Gamma=1}$ is a total order.



\subsection{Testing $r_1 \not \prec_\Gamma r_2$ using matched observational
  studies}
\label{sec:test-r_1-prec_g}

With the goal of selecting and ranking
treatment rules with unmeasured confounding in mind, in this section we
consider the easier but essential task of comparing the value of two treatment
rules, $r_1$ and $r_2$, under Rosenbaum's sensitivity model. This test
will then serve as the basic element of our procedures of selecting and
ranking among multiple treatment rules below. We will first
introduce the pair-matched design of an observational study and Rosenbaum's
sensitivity model, and then describe a studentized sensitivity
analysis proposed by \citet{fogarty2016studentized} that tests
Neyman's null hypothesis of average treatment effect being zero under
Rosenbaum's sensitivity model. This test can be immediately extended
to compare the value of treatment rules.

Suppose the observed data are $n$ pairs, $i = 1,2,...,n$, of two
subjects $j = 1,2$. These $n$ pairs are matched for observed
covariates $\bm X$ and within each pair, one subject is treated,
denoted $A_{ij} = 1$, and the other control, denoted $A_{ij} = 0$, so
that we have $\bm X_{i1} = \bm X_{i2}$ and $A_{i1} + A_{i2} = 1$ for all
$i$. In a sensitivity analysis, we may fail to match on
an unobserved confounder $U_{ij}$ and thus incur unmeasured
confounding bias.

\cite{Rosenbaum1987,Rosenbaum2002a} proposed a
one-parameter sensitivity model. Let $\mathcal{F} =
\{(Y_{ij}(0),Y_{ij}(1), \allowbreak \bm X_{ij}, U_{ij}), i = 1,\dotsc,n,
j = 1,2\}$ be the collection of all measured or unmeasured variables
other than the treatment assignment. Rosenbaum's sensitivity model
assumes that $\pi_i = P(A_{i1} = 1 | \mathcal{F})$ satisfies
\begin{equation}
  \frac{1}{1+\Gamma} \leq \pi_i \leq
  \frac{\Gamma}{1+\Gamma},~i =
  1,2,...,n.
  \label{eqn: rosenbaum model}
\end{equation}
When $\Gamma = 1$, this model asserts that $\pi_i = 1/2$ for all $i$
and thus every subject has equal probability to be assigned to
treatment or control (i.e.\ no unmeasured confounding). In general,
$\Gamma > 1$ controls the degree of departure from
randomization. \cite{Rosenbaum2002a,rosenbaum2011new} derived
randomization inference
based on signed score tests for Fisher's sharp null hypothesis that
$Y_{ij}(0) = Y_{ij}(1)$ for all $i,j$. The asymptotic properties of
these randomization tests are studied in
\citet{rosenbaum2004design,Rosenbaum2015} and
\citet{Zhao2018sens_value}.

In the context of comparing individualized treatment rules, Fisher's
sharp null hypothesis is no longer suitable because we expect to have
(and indeed are tasked to find) heterogeneous treatment effect.  Recently,
\cite{fogarty2016studentized}
developed a valid studentized test for Neyman's null hypothesis that
the average treatment effect is equal to zero, $(2n)^{-1}\sum_{ij}
Y_{ij}(1) - Y_{ij}(0) = 0$, under Rosenbaum's sensitivity model. We
briefly describe Fogarty's test. Let $D_i$ denote the
treated-minus-control difference in the $i^{th}$ matched pair,
$D_i = (A_{i1} - A_{i2}) (Y_{i1} - Y_{i2})$. Fix the
sensitivity parameter $\Gamma$ and define
\[
  D_{i, \Gamma} = D_i - \left(\frac{\Gamma - 1}{\Gamma +1}
  \right)|D_i|,~\overline{D}_{\Gamma} = \frac{1}{n}
  \sum_{i=1}^n D_{i,\Gamma},~\text{and}~
  \text{se}(\overline{D}_\Gamma)^2 = \frac{1}{n(n - 1)} \sum_{i = 1}^n (D_{i,
    \Gamma} - \overline{D}_{\Gamma})^2.
\]
\citet{fogarty2016studentized} showed that the one-sided student-$t$
test that rejects Neyman's hypothesis when
\[
  \frac{\overline{D}_\Gamma}{\text{se}(\overline{D}_\Gamma)}
  > \Phi^{-1}(1 - \alpha)
\]
is asymptotically valid with level $\alpha$ under Rosenbaum's
sensitivity model \eqref{eqn: rosenbaum
  model} and mild regularity
conditions. This test can be easily extended to test the null that the
average treatment effect is no greater than $\delta$ by replacing
$D_i$ with $D_i - \delta$.
\citet{fogarty2016studentized} also provided a
randomization-based reference distribution in addition to the
large-sample normal approximation.


The above test for the average treatment effect can be readily
extended to comparing treatment rules. Recall that equation
\eqref{eq:value-diff} implies the value difference of two rules $r_1$
and $r_2$ is a weighted difference of two conditional average
treatment effects on the set $r_1 > r_2$ and $r_2 > r_1$. When the two
rules are nested (without loss of generality assume $r_2 \ge r_1$),
testing the null hypothesis that $r_1 \not \prec_{\Gamma} r_2$
is equivalent to testing a Neyman-type hypothesis $\mathbb{E}[Y(1) - Y(0) |
r_2 > r_1] \le 0$ under the $\Gamma$-sensitivity model. We can simply
apply Fogarty's test to the matched pairs (indexed by $i$) that satisfy
$r_2(\bm X_{i1}) > r_1(\bm X_{i1})$. When the two rules are not nested, we can
flip the sign of $D_i$ for those $i$ such that $r_2(\bm X_{i1}) <
r_1(\bm X_{i1})$ and then apply Fogarty's test. In summary, to test the
null hypothesis that $r_1 \not \prec_{\Gamma} r_2$, we can
simply apply Fogarty's test to $\{D_i \cdot [r_2(\bm X_{i1}) - r_1(\bm X_{i1})],~\text{for}~i~\text{such
  that}~r_1(\bm X_{i1}) \ne r_2(\bm X_{i1})\}$. To test the hypothesis $r_1 \not
\prec_{\Gamma,\delta} r_2$, we can use Fogarty's test for the average
treatment effect no greater than $\delta \cdot (n / m)$ where $m =
\big|\{i:\,r_1(\bm X_{i1}) \ne r_2(\bm X_{i2})\}\big|$.

\subsection{Sensitivity value of treatment rule comparison}
A hallmark of Rosenbaum's sensitivity analysis framework is its tipping-point
analysis, and that extends to the comparison of treatment rules. When
testing $r_1 \not \prec_\Gamma r_2$ with a series of
$\Gamma$, there exists a smallest $\Gamma$ such that the null
hypothesis cannot be rejected, that is, we are no longer confidence
that $r_1$ is dominated by $r_2$ in that $\Gamma$-sensitivity
model. This tipping point is commonly referred to as the
\emph{sensitivity value} \citep{Zhao2018sens_value}. Formally, we
define the sensitivity value for $r_1 \prec r_2$ as
\begin{equation*}
  \begin{split}
    \Gamma_{\alpha}^\ast(r_1 \prec r_2) = \inf\{\Gamma \geq 1~:
    &~\text{The hypothesis}~V(r_1)
    \ge V(r_2) \mbox{ cannot be rejected} \\
    &\mbox{ at level $\alpha$ under the
      $\Gamma$-sensitivity model} \}.
  \end{split}
\end{equation*}
Let $r_0$ be the null treatment rule (for example, assigning control
to the entire population). The sensitivity $\Gamma_{\alpha}^\ast(r_0 \prec r_1)$ is further abbreviated as $\Gamma_{\alpha}^\ast(r_1)$.

\citet{Zhao2018sens_value} studied the asymptotic properties of the
sensitivity value when testing Fisher's sharp null hypothesis using a
class of signed core statistics. Below, we will give the asymptotic
distribution of $\Gamma_{\alpha}^\ast(r_1 \prec r_2)$ using Fogarty's
test as described in the last section. The result will be stated in
terms of a transformation of the sensitivity value,
\[
  \kappa^\ast_{\alpha}(r_1 \prec r_2) = \frac{\Gamma_{\alpha}^\ast(r_1 \prec
    r_2) - 1}{\Gamma_{\alpha}^\ast(r_1 \prec r_2) + 1}.
\]
Note that $\Gamma^\ast = 1$ is transformed to $\kappa^\ast = 0$ and
$0 \le \kappa^{\ast} < 1$.

\begin{proposition}
  Assume the treatment rules are nested, $r_1(x) \le r_2(x)$, and let $\mathcal{I}$ be the set of indices $i$ where $r_1(\bm X_{i1}) <
  r_2(\bm X_{i2})$. Assuming the moments of $|D_i|$ exist and
  $\mathbb{E}[D_i \mid r_1 < r_2] > 0$, then
  \begin{equation} \label{eq:sen-value-asymp}
    \begin{split}
      \sqrt{|\mathcal{I}|}\left(\kappa_{\alpha}^\ast(r_1 \prec r_2) -
        \frac{\mathbb{E}[D_i \mid r_1 < r_2]}{\mathbb{E}[|D_i| \mid
          r_1 < r_2]}\right) \overset{d}{\to}
      \text{N}\left(z_{\alpha} \mu,
        ~\sigma^2\right),~\text{as}~ |\mathcal{I}| \to \infty, \\
    \end{split}
  \end{equation}
  where $z_\alpha$ is the upper-$\alpha$ quantile of the standard normal
  distribution and the parameters $\mu$ and $\sigma^2$ depend on the
  distribution of $D_i$ (the expressions can be found in the Appendix).
  \label{thm: asymp kappa}
\end{proposition}

The proof of this proposition can be found in the Appendix. When the
treatment rules are not nested, one can simply replace $D_i$ with $D_i
[r_2(\bm X_{i1}) - r_1(\bm X_{i1})]$ and the condition $r_1 < r_2$
with $r_1 \ne r_2$ in the proposition statement. The asymptotic
distribution of $\Gamma^{\ast}_{\alpha}(r_1 \succ r_2)$ can be found by the delta
method and we omit further detail.

The asymptotic distribution in \eqref{eq:sen-value-asymp} is similar to
the one obtained in \citet[Thm.\ 1]{Zhao2018sens_value}. When the
treatment rules are nested $r_1 \le r_2$ and $|\mathcal{I}|
\to \infty$, the sensitivity value converges to a number that depends
on the distribution of $D_i$,
\[
  \Gamma_{\alpha}^{*}(r_1 \prec r_2) \overset{p}{\to}
  \frac{\mathbb{E}[|D_i| \mid r_1 < r_2] + \mathbb{E}[D_i \mid r_1 <
    r_2]}{\mathbb{E}[|D_i| \mid r_1 < r_2] -
    \mathbb{E}[D_i \mid r_1 < r_2]}.
\]
The limit on the right hand side is the \emph{design sensitivity}
\citep{rosenbaum2004design} of Fogarty's test for comparing the
treatment rules. As the sample size converge to infinity, the power of
Fogarty's test converges to $1$ at $\Gamma$ smaller than the design
sensitivity and to $0$ at $\Gamma$ larger than the design
sensitivity. The normal distribution in \eqref{eq:sen-value-asymp}
further approximates of the finite-sample behavior of the
sensitivity value and can be used to compute the power of a
sensitivity analysis by the fact that rejecting $r_1 \not
\prec_{\Gamma} r_2$ at level $\alpha$ is equivalent to
$\Gamma_{\alpha}^{*}(r_1 \prec_{\Gamma} r_2) \ge \Gamma$.


\section{Selecting and ranking treatment rules}
\label{sec:multiple}

Next we consider the problem of comparing multiple treatment rules
with unmeasured confounding. To this end, we need to define the goal
and the statistical error we
would like to control. A problem related to this is the selecting and
ordering of multiple subpopulations
\citep{gupta1979multiple,gibbons1999selecting}, for example, given $K$
independent measurements $Y_i \sim \mathrm{N}(\mu_i,1)$ where $\mu_i$
is some characteristic of the $i$-th
subpopulation. When comparing $\mu_i$, there are many goals
we can define. In fact, \citet[p.\ 4]{gibbons1999selecting} gave a
list of 7 possible goals for ranking and selection of subpopulations and
considered them in the rest of their book. We believe at least $3$ out
of their $7$ goals have practically meaningful counterparts in
comparing treatment rules. Given $K+1$ treatment rules, $\mathcal{R} =
\{r_0, r_1, \dotsc, r_K\}$, we may ask, in terms of their values,
\begin{enumerate}
\item What is the ordering of all the treatment rules?
\item Which treatment rule is the best?
\item Which treatment rule(s) are better than the null/control $r_0$?
\end{enumerate}
In a randomized experiment or an observational study with no
unmeasured confounding, it may be possible to obtain estimates of the
value that are jointly asymptotically normal and then directly use the
methods in \citet{gibbons1999selecting}. However, as discussed in
\Cref{sec:2}, this no longer
applies when there is unmeasured confounding because the value function
may only be partially identified.

\subsection{Defining the inferential goals}

When there is unmeasured confounding, the three goals above need to be
modified because the value
function only defines a partial order among the treatment rules
(\Cref{def: prec_gamma}). We make the following definitions
\begin{definition}
  In the $\Gamma$-sensitivity model, the \emph{maximal rules} in $\mathcal{R}$
  are the ones not dominated by any other rule,
  \[
    \mathcal{R}_{\text{max},\Gamma} = \{r_i\,:\,r_i \not \prec_{\Gamma}
    r_j,~\forall j\}.
  \]
  The \emph{positive rules} are the ones which dominate the control and the
  \emph{null rules} are the ones which don't dominate the control,
  \[
    \mathcal{R}_{\text{pos},\Gamma} = \{r_i\,:\,r_0 \prec_{\Gamma} r_i
    \},~\mathcal{R}_{\text{nul},\Gamma} = \mathcal{R} \setminus
    \mathcal{R}_{\text{pos},\Gamma}.
  \]
\end{definition}
The maximal set $\mathcal{R}_{\text{max},\Gamma}$ and the null set
$\mathcal{R}_{\text{nul},\Gamma}$ are always non-empty (the latter is
because $r_0 \in \mathcal{R}_{\text{nul},\Gamma}$), become
larger as $\Gamma$ increases, and in general become the full set
$\mathcal{R}$ as $\Gamma \to \infty$.

In the rest of this section, we will consider the following
three statistical problems: for some pre-specified significance level
$\alpha > 0$,
\begin{enumerate}
\item Can we give a set of ordered pairs of treatment rules,
  $\hat{\mathcal{O}}_{\Gamma} \subset
  \{(r_i,r_j),\,i,j=0,\dotsc,K,\,i\ne j\}$,
  such that the probability that all the orderings are correct is at
  least $1 - \alpha$, that is, $\mathbb{P}(r_i \prec_{\Gamma}
  r_j,\,\forall (r_i,r_j) \in \hat{\mathcal{O}}_{\Gamma}) \ge 1-\alpha$?
\item Can we construct a subset of treatment rules,
  $\hat{\mathcal{R}}_{\text{max},\Gamma}$, such that the probability that it
  contains all maximal rules is at least $1-\alpha$, that
  is, $\mathbb{P}(\mathcal{R}_{\text{max},\Gamma} \subseteq
  \hat{\mathcal{R}}_{\text{max},\Gamma}) \ge 1-\alpha$?
\item Can we construct a subset of treatment rules,
  $\hat{\mathcal{R}}_{\text{pos},\Gamma}$, such that the probability
  that it does not cover any null rule is at least $1 - \alpha$, that is, $\mathbb{P}(\hat{\mathcal{R}}_{\text{pos},\Gamma} \cap
  \mathcal{R}_{\text{null},\Gamma} = \not\emptyset) \ge 1-\alpha$?
\end{enumerate}

Next, we will propose strategies to achieve the above statistical
goals based on the test of two treatment rules with unmeasured
confounding described in \Cref{sec:test-r_1-prec_g}.

\subsection{Goal 1: Ordering the treatment rules}
\label{subsec: ranking}

To start with, let's consider the first goal---ordering the treatment
rules, as the statistical inference is more straightforward. It is the
same as the multiple testing problem where we would like to control
the family-wise error rate (FWER) for
the collection of hypotheses, $\{H_{ij}\,:\,r_i \not \prec_{\Gamma}
r_j,\,i,j=0,\dotsc,K,\,i\ne j\}$. In principle, we can apply any
multiple testing procedure that controls the FWER. A simple example is
Bonferroni's correction for all the $K(K-1)$ tests.

In sensitivity analysis problems, we can often greatly improve the
statistical power by reducing the number of tests using a planning
sample \citep{heller2009split,zhao2018cross}. This is because
Rosenbaum's sensitivity analysis considers the worst case scenario and
is generally conservative when $\Gamma > 1$. The planning sample can
be further used to order the hypotheses so we can sequentially test
them, for example, using a fixed sequence testing procedure
\citep{koch1996statistical,westfall2001optimally}.

There are many possible ways to screen out, order, and then test the
hypotheses. Here we demonstrate one possibility:
\begin{itemize}
\item[Step 1:] Split the data into two parts. The first part
  is used for planning and the second part for testing.
\item[Step 2:] For every pair of treatment rules $(r_i,r_j)$, use the
  planning sample to estimate population parameters in the asymptotic
  distribution of the sensitivity value \eqref{eq:sen-value-asymp}.
\item[Step 3:] Compute the approximate power of testing $H_{ij}:~r_i
  \not \prec_{\Gamma} r_j$ in the testing sample using
  \eqref{eq:sen-value-asymp}. Order the hypotheses by the estimated
  power, from highest to lowest.
\item[Step 4:] Sequentially test the ordered hypotheses using the
  testing sample at level $\alpha$, until one hypothesis is rejected.
\item[Step 5:] Output a Hasse diagram of the treatment rules by using all
  the rejected hypotheses.
\end{itemize}

A Hasse diagram is an informative graph to represent a partial order
(in our case, $\prec_{\Gamma}$). In this diagram, each treatment rule
is represented by a vertex and an edge goes upward from rule
$r_i$ to rule $r_j$ if $r_i \prec_{\Gamma} r_j$ and there exists no $r_k$
such that $r_i \prec_{\Gamma} r_k$ and $r_k \prec_{\Gamma} r_j$.

Due to transitivity of a partial order, an
upward path from $r_i$ to
$r_j$ in the Hasse diagram {(for example $r_0$ to $r_3$ in Figure 1,
  $\Gamma = 1.3$)} indicates that $r_i \prec_{\Gamma} r_j$, even if we could not directly
reject
$r_i \not \prec_{\Gamma} r_j$ in Step 4. The next proposition shows
that the above multiple testing procedure also controls the FWER for all the
apparent and implied orders represented by the Hasse diagram.

\begin{proposition}
  Let $\hat{\mathcal{O}}_{\Gamma} \subset \{(r_i,r_j),\,i\ne j\}$ be
  a random set of ordered treatment rules obtained using the
  procedure above or any other multiple testing procedure. Let
  \[
    \hat{\mathcal{O}}_{\Gamma,\text{ext}} = \hat{\mathcal{O}}_{\Gamma}
    \bigcup \{(r_i,r_j)\,:\,\exists\, k_1,\dotsc,k_m~\text{such that}~
    (r_i,r_{k_1}),(r_{k_1},r_{k_2}),\dotsc,(r_{k_m},r_j)
    \in \hat{\mathcal{O}}_{\Gamma}\}
  \]
  be the extended set implied from the Hasse diagram. Then FWER with
  respect to $\hat{\mathcal{O}}_{\Gamma,\text{ext}}$ is the same as FWER with
  respect to $\hat{\mathcal{O}}_{\Gamma}$:
  \[
    \mathbb{P}(r_i \prec_{\Gamma}
    r_j,\,\forall (r_i,r_j) \in \hat{\mathcal{O}}_{\Gamma}) =
    \mathbb{P}(r_i \prec_{\Gamma}
    r_j,\,\forall (r_i,r_j) \in \hat{\mathcal{O}}_{\Gamma,\text{ext}}).
  \]
\end{proposition}
\begin{proof}
  We show the two events are equivalent. The $\subseteq$ direction is
  trivial. For $\supseteq$, notice that any false
  positive in $\hat{\mathcal{O}}_{\Gamma,\text{ext}}$, say $r_i
  \prec_{\Gamma} r_j$ implies that there is at least one false
  positive along the path from $r_i$ to $r_j$, that is, there is at
  least one false positive among $r_i \prec_{\Gamma} r_{k_1},
  r_{k_1}\prec_{\Gamma} r_{k_2},\dotsc,r_{k_m} \prec_{\Gamma}
  r_j$, which are all in $\hat{\mathcal{O}}_{\Gamma}$. Thus, any false
  positive in $\hat{\mathcal{O}}_{\Gamma,\text{ext}}$ implies that there is
  also at least one false positive in $\hat{\mathcal{O}}_{\Gamma}$.
\end{proof}

We illustrate the proposed method using the malaria
dataset. We first use half of the data to
estimate the population parameters in \eqref{eq:sen-value-asymp} for
each pair of treatment rules $(r_i, r_j)$. For every value of
$\Gamma$, we use
\eqref{eq:sen-value-asymp} to compute the asymptotic power for the
test of $H_{ij}:r_i \not \prec_{\Gamma} r_j$ using the other half of
the data. We then order the hypotheses by the estimated power, from
the highest to the lowest. In the
malaria example, when $\Gamma = 1$, the order is
\[
  H_{01}, H_{02}, H_{03}, H_{04}, H_{05}, H_{13}, H_{12}, H_{14}, H_{15}, H_{23}, \dotsc.
\]
When $\Gamma = 2$, the order becomes
\[
  H_{02}, H_{01}, H_{03}, H_{04}, H_{05}, H_{12}, H_{13}, H_{14}, H_{15}, H_{45}, \dotsc.
\]
Finally we follow Steps 4 and 5 above. We obtained Hasse diagrams for
a variety of $\Gamma$, which are shown in \Cref{fig:
  hasse malaria hybrid}. As a baseline for comparison, \Cref{fig:
  malaria hasse Bonferroni} shows the Hasse diagrams obtained by a
simple Bonferroni
adjustment for all $K(K-1) = 30$ hypotheses using all the data. Although only half of the data is used to test, ordering the hypotheses not only
identified all the discoveries that the Bonferroni procedure
identified, but also made one extra discovery when $\Gamma = 1.3$,
$1.5$, $2.5$, $3.5$, and $4$, and
two extra discoveries when $\Gamma = 1$, $1.8$, $2$, and $3$.

\begin{figure}[h] \centering
  \subfigure{
    \begin{tikzpicture}[scale=0.35]
      \node (5) at (-2,0) {$r_5$};
      \node (4) at (0,0) {$r_4$};
      \node (3) at (2,0) {$r_3$};
      \node (2) at (0,-2) {$r_2$};
      \node (1) at (0, -4) {$r_1$};
      \node (0) at (0, -6) {$r_0$};
      \draw (5) -- (2) -- (4) -- (2) -- (3) -- (2) -- (1) -- (0);
      \node [below of=0] {
        \begin{tabular}{c}
          $\Gamma = 1$ \\
          $|\hat{\mathcal{O}}| = 12$
        \end{tabular}
      };
    \end{tikzpicture}}
  \qquad
  \subfigure{
    \begin{tikzpicture}[scale=0.35]
      \node (3) at (0,2) {$r_3$};
      \node (4) at (-2,0) {$r_4$};
      \node (2) at (0,0) {$r_2$};
      \node (5) at (2,0) {$r_5$};
      \node (1) at (0, -2) {$r_1$};
      \node (0) at (0, -4) {$r_0$};
      \draw (3) -- (2) -- (1) -- (5) -- (1) -- (4) -- (1) -- (0);
      \node [below of=0] {
        \begin{tabular}{c}
          $\Gamma = 1.3$ \\
          $|\hat{\mathcal{O}}| = 10$
        \end{tabular}
      };
    \end{tikzpicture}}
  \qquad
  \subfigure{
    \begin{tikzpicture}[scale=0.35]
      \node (5) at (-3,0) {$r_5$};
      \node (4) at (-1,0) {$r_4$};
      \node (3) at (1,0) {$r_3$};
      \node (2) at (3,0) {$r_2$};
      \node (1) at (0, -2) {$r_1$};
      \node (0) at (0, -4) {$r_0$};
      \draw (5) -- (1) -- (4) -- (1) -- (3) -- (1) -- (2) -- (1) -- (0);
      \node [below of=0] {
        \begin{tabular}{c}
          $\Gamma = 1.5$ \\
          $|\hat{\mathcal{O}}| = 9$
        \end{tabular}
      };
    \end{tikzpicture}}
  \qquad
  \subfigure{
    \begin{tikzpicture}[scale=0.35]
      \node (2) at (-2,0) {$r_2$};
      \node (3) at (0,0) {$r_3$};
      \node (4) at (2,0) {$r_4$};
      \node (1) at (0,-2) {$r_1$};
      \node (5) at (2, -2) {$r_5$};
      \node (0) at (0, -4) {$r_0$};
      \draw (2) -- (1) -- (3) -- (1) -- (4) -- (1) -- (0) -- (5);
      \node [below of=0] {
        \begin{tabular}{c}
          $\Gamma = 1.8$ \\
          $|\hat{\mathcal{O}}| = 8$
        \end{tabular}
      };
    \end{tikzpicture}} \\ \vspace{0.5cm}
  \subfigure{
    \begin{tikzpicture}[scale=0.35]
      \node (2) at (-2,0) {$r_2$};
      \node (3) at (2,0) {$r_3$};
      \node (1) at (0,-2) {$r_1$};
      \node (4) at (-2,-2) {$r_4$};
      \node (5) at (2, -2) {$r_5$};
      \node (0) at (0, -4) {$r_0$};
      \draw (2) -- (1) -- (3) -- (1) -- (0) -- (4) -- (0) -- (5);
      \node [below of=0] {
        \begin{tabular}{c}
          $\Gamma = 2$ \\
          $|\hat{\mathcal{O}}| = 7$
        \end{tabular}
      };
    \end{tikzpicture}}
  \qquad
  \subfigure{
    \begin{tikzpicture}[scale=0.33]
      \node (2) at (0,0) {$r_2$};
      \node (1) at (0,-2) {$r_1$};
      \node (3) at (-2,-2) {$r_3$};
      \node (4) at (2,-2) {$r_4$};
      \node (5) at (4, -2) {$r_5$};
      \node (0) at (0, -4) {$r_0$};
      \draw (2) -- (1) -- (0) -- (3) -- (0) -- (4) -- (0) -- (5);
      \node [below of=0] {
        \begin{tabular}{c}
          $\Gamma = 2.5$ \\
          $|\hat{\mathcal{O}}| = 6$
        \end{tabular}
      };
    \end{tikzpicture}}
  \qquad
  \subfigure{
    \begin{tikzpicture}[scale=0.33]
      \node (1) at (-4,-2) {$r_1$};
      \node (2) at (-2,-2) {$r_2$};
      \node (3) at (0,-2) {$r_3$};
      \node (4) at (2,-2) {$r_4$};
      \node (5) at (4, -2) {$r_5$};
      \node (0) at (0, -4) {$r_0$};
      \draw (1) -- (0) -- (2) -- (0) -- (3) -- (0) -- (4) -- (0) -- (5);
      \node [below of=0] {
        \begin{tabular}{c}
          $\Gamma = 3.0$ \\
          $|\hat{\mathcal{O}}| = 5$
        \end{tabular}
      };
    \end{tikzpicture}} \\ \vspace{0.5cm}
  \qquad
  \subfigure{
    \begin{tikzpicture}[scale=0.33]
      \node (1) at (0,-2) {$r_1$};
      \node (2) at (-2,-2) {$r_2$};
      \node (3) at (2,-2) {$r_3$};
      \node (0) at (0, -4) {$r_0$};
      \node (4) at (2, -4) {$r_4$};
      \node (5) at (4, -4) {$r_5$};
      \draw (1) -- (0) -- (2) -- (0) -- (3) -- (0);
      \draw (5);
      \draw (4);
      \node [below of=0] {
        \begin{tabular}{c}
          $\Gamma = 3.5$ \\
          $|\hat{\mathcal{O}}| = 3$
        \end{tabular}
      };
    \end{tikzpicture}}
  \qquad
  \subfigure{
    \begin{tikzpicture}[scale=0.33]
      \node (1) at (0,-2) {$r_1$};
      \node (2) at (-4,-2) {$r_2$};
      \node (0) at (-2, -4) {$r_0$};
      \node (3) at (0,-4) {$r_3$};
      \node (4) at (2, -4) {$r_4$};
      \node (5) at (4, -4) {$r_5$};
      \draw (2) -- (0) -- (1);
      \draw (3);
      \draw (5);
      \draw (4);
      \node [below of=0] {
        \begin{tabular}{c}
          $\Gamma = 4$ \\
          $|\hat{\mathcal{O}}| = 2$
        \end{tabular}
      };
    \end{tikzpicture}}
  \qquad
  \subfigure{
    \begin{tikzpicture}[scale=0.33]
      \node (1) at (-6,-4) {$r_1$};
      \node (2) at (-4,-4) {$r_2$};
      \node (0) at (-2, -4) {$r_0$};
      \node (3) at (0,-4) {$r_3$};
      \node (4) at (2, -4) {$r_4$};
      \node (5) at (4, -4) {$r_5$};
      \node [below of=0] {
        \begin{tabular}{c}
          $\Gamma = 6$ \\
          $|\hat{\mathcal{O}}| = 0$
        \end{tabular}
      };
    \end{tikzpicture}}

  \caption{Malaria example: Hasse diagrams obtained using sample-splitting and fixed
    sequence testing; $\alpha = 0.1$.}
  \label{fig: hasse malaria hybrid}
\end{figure}
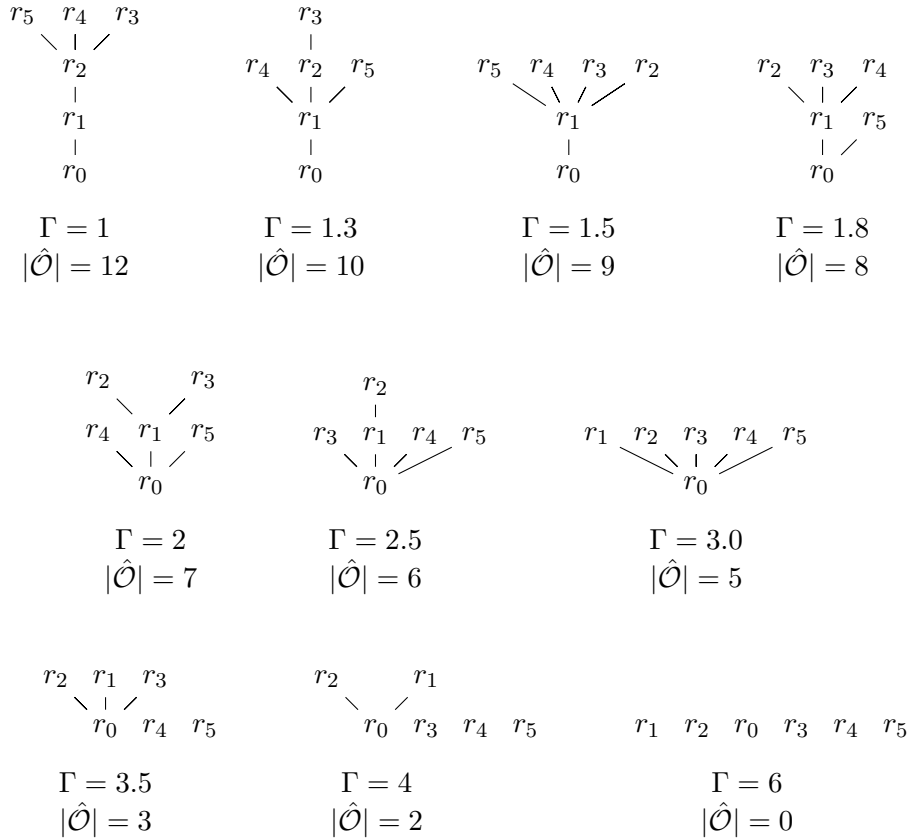

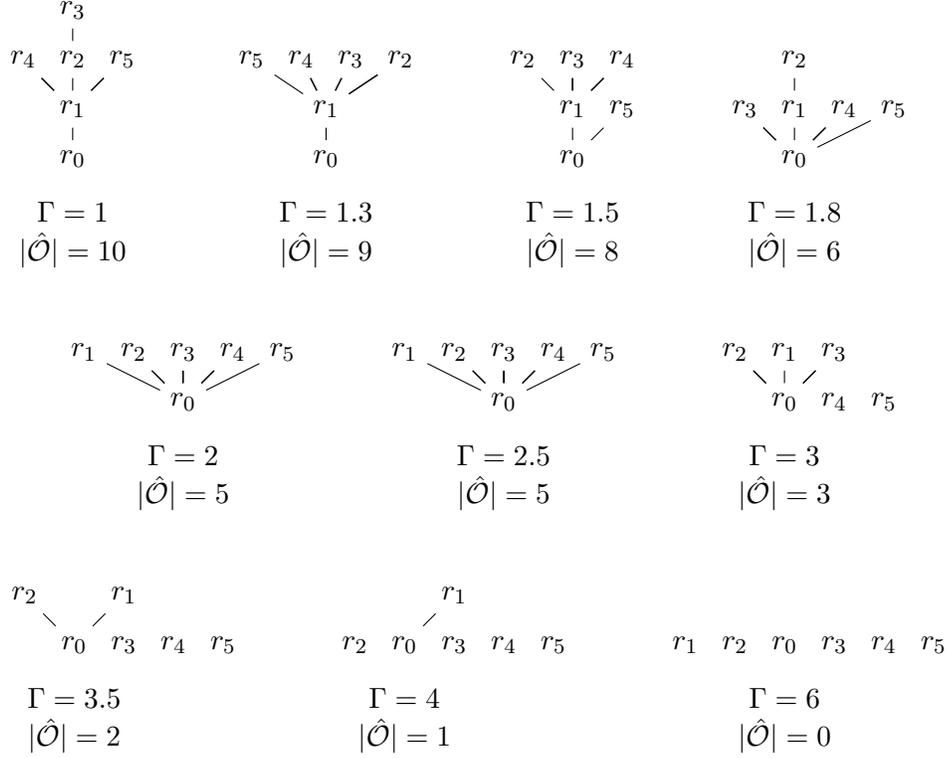
\begin{figure}
  \centering

  \subfigure{
    \begin{tikzpicture}[scale=0.33]
      \node (3) at (0,2) {$r_3$};
      \node (4) at (-2,0) {$r_4$};
      \node (2) at (0,0) {$r_2$};
      \node (5) at (2,0) {$r_5$};
      \node (1) at (0, -2) {$r_1$};
      \node (0) at (0, -4) {$r_0$};
      \draw (3) -- (2) -- (1) -- (5) -- (1) -- (4) -- (1) -- (0);
      \node [below of=0] {
        \begin{tabular}{c}
          $\Gamma = 1$ \\
          $|\hat{\mathcal{O}}| = 10$
        \end{tabular}
      };
    \end{tikzpicture}}
  \qquad
  \subfigure{
    \begin{tikzpicture}[scale=0.33]
      \node (5) at (-3,0) {$r_5$};
      \node (4) at (-1,0) {$r_4$};
      \node (3) at (1,0) {$r_3$};
      \node (2) at (3,0) {$r_2$};
      \node (1) at (0, -2) {$r_1$};
      \node (0) at (0, -4) {$r_0$};
      \draw (5) -- (1) -- (4) -- (1) -- (3) -- (1) -- (2) -- (1) -- (0);
      \node [below of=0] {
        \begin{tabular}{c}
          $\Gamma = 1.3$ \\
          $|\hat{\mathcal{O}}| = 9$
        \end{tabular}
      };
    \end{tikzpicture}}
  \qquad
  \subfigure{
    \begin{tikzpicture}[scale=0.33]
      \node (2) at (-2,0) {$r_2$};
      \node (3) at (0,0) {$r_3$};
      \node (4) at (2,0) {$r_4$};
      \node (1) at (0,-2) {$r_1$};
      \node (5) at (2, -2) {$r_5$};
      \node (0) at (0, -4) {$r_0$};
      \draw (2) -- (1) -- (3) -- (1) -- (4) -- (1) -- (0) -- (5);
      \node [below of=0] {
        \begin{tabular}{c}
          $\Gamma = 1.5$ \\
          $|\hat{\mathcal{O}}| = 8$
        \end{tabular}
      };
    \end{tikzpicture}}
  \qquad
  \subfigure{
    \begin{tikzpicture}[scale=0.33]
      \node (2) at (0,0) {$r_2$};
      \node (1) at (0,-2) {$r_1$};
      \node (3) at (-2,-2) {$r_3$};
      \node (4) at (2,-2) {$r_4$};
      \node (5) at (4, -2) {$r_5$};
      \node (0) at (0, -4) {$r_0$};
      \draw (2) -- (1) -- (0) -- (3) -- (0) -- (4) -- (0) -- (5);
      \node [below of=0] {
        \begin{tabular}{c}
          $\Gamma = 1.8$ \\
          $|\hat{\mathcal{O}}| = 6$
        \end{tabular}
      };
    \end{tikzpicture}} \\ \vspace{0.5cm}
  \qquad
  \subfigure{
    \begin{tikzpicture}[scale=0.33]
      \node (1) at (-4,-2) {$r_1$};
      \node (2) at (-2,-2) {$r_2$};
      \node (3) at (0,-2) {$r_3$};
      \node (4) at (2,-2) {$r_4$};
      \node (5) at (4, -2) {$r_5$};
      \node (0) at (0, -4) {$r_0$};
      \draw (1) -- (0) -- (2) -- (0) -- (3) -- (0) -- (4) -- (0) -- (5);
      \node [below of=0] {
        \begin{tabular}{c}
          $\Gamma = 2$ \\
          $|\hat{\mathcal{O}}| = 5$
        \end{tabular}
      };
    \end{tikzpicture}}
  \qquad
  \subfigure{
    \begin{tikzpicture}[scale=0.33]
      \node (1) at (-4,-2) {$r_1$};
      \node (2) at (-2,-2) {$r_2$};
      \node (3) at (0,-2) {$r_3$};
      \node (4) at (2,-2) {$r_4$};
      \node (5) at (4, -2) {$r_5$};
      \node (0) at (0, -4) {$r_0$};
      \draw (1) -- (0) -- (2) -- (0) -- (3) -- (0) -- (4) -- (0) -- (5);
      \node [below of=0] {
        \begin{tabular}{c}
          $\Gamma = 2.5$ \\
          $|\hat{\mathcal{O}}| = 5$
        \end{tabular}
      };
    \end{tikzpicture}
  }
  \qquad
  \subfigure{
    \begin{tikzpicture}[scale=0.33]
      \node (1) at (0,-2) {$r_1$};
      \node (2) at (-2,-2) {$r_2$};
      \node (3) at (2,-2) {$r_3$};
      \node (0) at (0, -4) {$r_0$};
      \node (4) at (2, -4) {$r_4$};
      \node (5) at (4, -4) {$r_5$};
      \draw (1) -- (0) -- (2) -- (0) -- (3) -- (0);
      \node [below of=0] {
        \begin{tabular}{c}
          $\Gamma = 3$ \\
          $|\hat{\mathcal{O}}| = 3$
        \end{tabular}
      };
    \end{tikzpicture}} \\ \vspace{0.5cm}
  \qquad
  \subfigure{
    \begin{tikzpicture}[scale=0.33]
      \node (1) at (0,-2) {$r_1$};
      \node (2) at (-4,-2) {$r_2$};
      \node (0) at (-2, -4) {$r_0$};
      \node (3) at (0,-4) {$r_3$};
      \node (4) at (2, -4) {$r_4$};
      \node (5) at (4, -4) {$r_5$};
      \draw (2) -- (0) -- (1);
      \node [below of=0] {
        \begin{tabular}{c}
          $\Gamma = 3.5$ \\
          $|\hat{\mathcal{O}}| = 2$
        \end{tabular}
      };
    \end{tikzpicture}
  }
  \qquad
  \subfigure{
    \begin{tikzpicture}[scale=0.33]
      \node (1) at (0,-2) {$r_1$};
      \node (0) at (-2, -4) {$r_0$};
      \node (2) at (-4,-4) {$r_2$};
      \node (3) at (0,-4) {$r_3$};
      \node (4) at (2, -4) {$r_4$};
      \node (5) at (4, -4) {$r_5$};
      \draw (1) -- (0);
      \node [below of=0] {
        \begin{tabular}{c}
          $\Gamma = 4$ \\
          $|\hat{\mathcal{O}}| = 1$
        \end{tabular}
      };
    \end{tikzpicture}
  }
  \qquad
  \subfigure{
    \begin{tikzpicture}[scale=0.33]
      \node (1) at (-6,-4) {$r_1$};
      \node (2) at (-4,-4) {$r_2$};
      \node (0) at (-2, -4) {$r_0$};
      \node (3) at (0,-4) {$r_3$};
      \node (4) at (2, -4) {$r_4$};
      \node (5) at (4, -4) {$r_5$};
      \node [below of=0] {
        \begin{tabular}{c}
          $\Gamma = 6$ \\
          $|\hat{\mathcal{O}}| = 0$
        \end{tabular}
      };
    \end{tikzpicture}
  }
  \caption{Malaria example: Hasse diagrams obtained using
    Bonferroni's adjustment; $\alpha = 0.1$.}
  \label{fig: malaria hasse Bonferroni}
\end{figure}

\subsection{Goal 2: Selecting the best rules}

Next we consider constructing a set that covers all the maximal
rules. Our proposal is based on the following observation: if the
hypothesis $r_i \not \prec_{\Gamma} r_j$ can be rejected, then $r_i$ is
unlikely a maximal rule. More precisely, because $r_i \in
\mathcal{R}_{\text{max},\Gamma}$ implies that $r_i \not \prec_{\Gamma}
r_j$ must be true, by the definition of the type I error of a
hypothesis test,
\[
  \mathbb{P}(r_i \not \prec_{\Gamma} r_j~\text{is rejected} \,|\,r_i \in
  \mathcal{R}_{\text{max},\Gamma}) \le \alpha.
\]
This suggests that we can derive a set of maximal rules from an
estimated set of partial orders:
\begin{equation} \label{eq:est-max}
  \hat{\mathcal{R}}_{\text{max}, \Gamma} = \{r_i\,:\,(r_i,r_j) \not \in
  \hat{\mathcal{O}}_{\Gamma},~\forall j\}.
\end{equation}
In other words, $\hat{\mathcal{R}}_{\text{max}, \Gamma}$ contains all
the ``leaves'' in the Hasse diagram of
$\hat{\mathcal{O}}_{\Gamma}$ (a leaf in the Hasse diagram is a vertex
who has no edge going upward). For example, in Figure \ref{fig: hasse
  malaria hybrid}, the leaves are $\{r_3, r_4, r_5\}$ when $\Gamma = 1.0$ and $\{r_2, r_3, r_4, r_5\}$ when $\Gamma = 1.5$.
Because $\big\{\mathcal{R}_{\text{max},\Gamma} \not \subseteq
\hat{\mathcal{R}}_{\text{max},\Gamma}\big\} = \big\{\exists\,
i \in \mathcal{R}_{\text{max},\Gamma}~\text{such that}~(r_i,r_j)\in
\hat{\mathcal{O}}_{\Gamma}~\text{for some}~j\big\}$, the estimated set
of maximal rules satisfies
$\mathbb{P}(\mathcal{R}_{\text{max},\Gamma} \not \subseteq
\hat{\mathcal{R}}_{\text{max},\Gamma}) \le \alpha$ as desired whenever
$\hat{\mathcal{O}}_{\Gamma}$ strongly controls the FWER at level $\alpha$.

Equation \eqref{eq:est-max} suggests that only one hypothesis $r_i
\not \prec_{\Gamma} r_j$ needs to be rejected in order to exclude
$r_i$ from $\hat{\mathcal{R}}_{\text{max}, \Gamma}$. This means that,
when the purpose is to select the maximal rules, we do not need to test
$r_i \not \prec_{\Gamma} r_j$ if another hypothesis $r_i \not
\prec_{\Gamma} r_k$ for some $k \ne j$ is already rejected. Therefore, we
can modify the procedure of finding $\hat{\Omega}_{\Gamma}$ to further
decrease the size of $\hat{\mathcal{R}}_{\text{max}, \Gamma}$ obtained
from \eqref{eq:est-max}. For example, in the five-step procedure
demonstrated in \Cref{subsec: ranking}, we can further replace Step 3 by:
\begin{itemize}
\item[Step 3':] After ordering the hypotheses in Step 3,
  remove any hypothesis $H_{ij}:\,r_i \prec_{\Gamma} r_j$ if
  there is already a hypothesis $H_{ik}$ appearing before $H_{ij}$ for
  some $k \ne j$.
\end{itemize}

Again we use the malaria example to illustrate the selection of best
treatment rules. As an example, when $\Gamma = 2.0$, Step 3' reduced the
original sequence of hypotheses to the following:
\[
  H_{02}, H_{12}, H_{45}, H_{35}, H_{53}, H_{21}.
\]
We used the hold-out samples to test the hypotheses sequentially at
level $\alpha = 0.1$ and stopped at $H_{45}$. Therefore, a level
$\alpha = 0.1$ confidence set of the set of maximal elements is
$\{r_2, r_3, r_4, r_5\}$ when $\Gamma = 2$. Table \ref{tbl: malaria
  max set result} lists the estimated maximal set
$\hat{\mathcal{R}}_{\text{max}, \Gamma}$ for $\Gamma = 1, 1.3, 1.5, 1.8, 2, \text{and }2.5$.

\begin{table}[h]
  \centering
  \caption{$\hat{\mathcal{R}}_{\text{max}, \Gamma}$ for different choices of $\Gamma$.}
  \label{tbl: malaria max set result}
  \begin{tabular}{cc|cc}
    \toprule
    $\Gamma$ &$\hat{\mathcal{R}}_{\text{max}, \Gamma}$ &$\Gamma$
    &$\hat{\mathcal{R}}_{\text{max}, \Gamma}$  \\
    \midrule
    1.0     & $\{r_3, r_4, r_5\}$       & 2.5 &$\{r_2, r_3, r_4, r_5\}$ \\
    1.3     & $\{r_3, r_4, r_5\}$       & 3.0 &$\{r_1, r_2, r_3, r_4, r_5\}$\\
    1.5     &  $\{r_2, r_3, r_4, r_5\}$ & 3.5 &$\{r_1, r_2, r_3, r_4, r_5\}$ \\
    1.8     &  $\{r_2, r_3, r_4, r_5\}$ & 4.0 &$\{r_1, r_2, r_3, r_4, r_5\}$ \\
    2.0     & $\{r_2, r_3, r_4, r_5\}$  & 6.0 &$\{r_0, r_1, r_2, r_3, r_4, r_5\}$\\
    \bottomrule
  \end{tabular}
\end{table}

\subsection{Goal 3: Selecting the positive rules}

Finally we consider how to select treatment rules that are better than
a control rule. This can also be transformed to a multiple testing
problem for the hypotheses
$H_{0i}:\,r_0\not\prec_{\Gamma}r_i,~i=1,\dotsc,K$. Let
$\hat{\mathcal{R}}_{\text{pos},\Gamma}$ be the collection of rejected hypotheses
following some multiple testing procedure. By definition of FWER,
$\mathbb{P}(\hat{\mathcal{R}}_{\text{pos},\Gamma} \cap
\mathcal{R}_{\text{nul},\Gamma}) \le \alpha$ if the multiple testing
procedure strongly controls FWER at level $\alpha$. As an example, one
can modify the procedure in \eqref{subsec: ranking} to select the
positive rules by only considering $H_{0i},~i=1,\dotsc,K$ in Step 3.

In practice, a small increase of the value function, though
statistically significant, may not justify a policy change. In this
case, it may be desirable to estimate the positive rules that dominate
the control rule by margin $\delta$,
$\mathcal{R}_{\text{pos},\Gamma,\delta} = \{r_i:\,r_0
\prec_{\Gamma,\delta} r_i\}$. To obtain an estimate of
$\mathcal{R}_{\text{pos},\Gamma,\delta}$, one can further modify the
procedure in \eqref{subsec: ranking} by replacing the hypothesis $H_{0i}:\,r_0
\not \prec_{\Gamma} r_i$ with the stronger $r_0
\not\prec_{\Gamma,\delta} r_i$.

\begin{table}[h]
  \centering
  \caption{Estimated positive rules $\hat{\mathcal{R}}_{\text{pos},\Gamma, \delta}$ for different choices of $\Gamma$ and $\delta$. 
  }
  \label{tbl: malaria null set result}
  \begin{tabular}{l|cccc}
    \toprule
    &$\Gamma = 1$ &$\Gamma = 1.3$ &$\Gamma = 1.5$&$\Gamma = 1.8$  \\
    \midrule
    $\delta = 0$ & $\{r_1, r_2, r_3, r_4, r_5\}$ &$\{r_1, r_2, r_3, r_4, r_5\}$& $\{r_1, r_2, r_3, r_4, r_5\}$&$\{r_1, r_2, r_3, r_4, r_5\}$ \\
    $\delta = 1$ & $\{r_1, r_2, r_3, r_4, r_5\}$ &$\{r_1, r_2, r_3, r_4, r_5\}$& $\{r_1, r_2, r_3, r_4, r_5\}$&$\{r_1, r_2, r_3, r_4, r_5\}$ \\
    $\delta = 2$ & $\{r_1, r_2, r_3, r_4, r_5\}$ &$\{r_1, r_2, r_3, r_4, r_5\}$& $\{r_1, r_2, r_3, r_4, r_5\}$&$\{r_1, r_2, r_3, r_4, r_5\}$ \\
    $\delta = 4$ & $\{r_1, r_2, r_3, r_4, r_5\}$ &$\{r_1, r_2, r_3, r_4, r_5\}$& $\{r_1, r_2, r_3, r_4, r_5\}$&$\{r_1, r_2, r_3, r_4, r_5\}$ \\
    $\delta = 6$ & $\{r_1, r_2, r_3, r_4, r_5\}$ &$\{r_1, r_2, r_3, r_4, r_5\}$& $\{r_2, r_3, r_4,r_5\}$&$\{r_2\}$ \\
    \midrule
    &$\Gamma = 2.0$ &$\Gamma = 2.5$&$\Gamma = 3.0$ &  \\
    \midrule
    $\delta = 0$  &$\{r_1, r_2, r_3, r_4, r_5\}$ & $\{r_1, r_2, r_3, r_4, r_5\}$& $\{r_1, r_2, r_3, r_4, r_5\}$ &  \\
    $\delta = 1$  & $\{r_1, r_2, r_3, r_4, r_5\}$ & $\{r_1, r_2, r_3, r_4, r_5\}$  & $\{r_1, r_2, r_3\}$ & \\
    $\delta = 2$  & $\{r_1, r_2, r_3, r_4, r_5\}$ & $\{r_1, r_2, r_3\}$& $\{r_1, r_2\}$ & \\
    $\delta = 4$  & $\{r_1, r_2, r_3\}$ & $\emptyset$ & $\emptyset$ & \\
    $\delta = 6$  & $\emptyset$ & $\emptyset$  & $\emptyset$ & \\
    \midrule
    & $\Gamma = 3.5$ &$\Gamma = 4.0$  &$\Gamma = 6.0$  &\\
    \midrule
    $\delta = 0$ & $\{r_1, r_2, r_3\}$ & $\{r_1, r_2\}$ & $\emptyset$  & \\
    $\delta = 1$ & $\{r_1, r_2\}$ & $\{r_1\}$ & $\emptyset$  \\
    $\delta = 2$ & $\emptyset$ & $\emptyset$ & $\emptyset$\\
    $\delta = 4$ & $\emptyset$  & $\emptyset$ & $\emptyset$\\
    $\delta = 6$ & $\emptyset$  & $\emptyset$ & $\emptyset$\\
    \bottomrule
  \end{tabular}
\end{table}

We construct $\hat{\mathcal{R}}_{\text{pos},\Gamma, \delta}$ with
various choices of $\Gamma$ and $\delta$ for the malaria example. In
this case, $\delta$ measures the decrease in the number of Plasmodium
falciparum parasites per milliliter of blood samples averaged over the
entire study samples. Table \ref{tbl: malaria null set result} gives a
summary of the results. As expected, the estimated set of positive
rules becomes smaller as $\Gamma$ or $\delta$ increases. We observe
that, although $r_1,r_2$---assigning treatment to those under $7$ and
$20$---are unlikely the optimal rules if there is no
unmeasured confounding (\Cref{tbl: malaria max set result}), they are
more robust to unmeasured confounding than the others, dominating the
control rule up till $\Gamma = 4.0$ (\Cref{tbl: malaria null set
  result}).

\section{Simulations}
\label{sec:simulations}

  We study and report the performance of three methods of selecting the
  positive rules $\mathcal{R}_{\text{pos},\Gamma, \delta}$ using
  numerical simulations in this section. Simulation results for
  selecting the maximal rules are reported in the Supplementary Materials.
We constructed $5$ or $10$ cohorts of data where the
treatment effect is constant within each cohort but different between
the cohorts. After matching, the treated-minus-control difference in
each cohort was normally distributed with mean
\begin{enumerate}
\item $\mu = (0.5, 0.25, 0.25, 0.15, 0.05)$,
\item $\mu = (0.5, 0.2, -1.0, 0.2, 0.5)$,
\item $\mu = (0.5, 0.5, 0.25, 0.25, 0.25, 0.25, 0.15, 0.15, 0.05,
  0.05)$,
\item $\mu = (0.5, 0.3, 0.2, 0.0, -1.0, -1.0, 0.5, 0.5, 1.0, 1.0)$.
\end{enumerate}
The size of each cohort was either $100$ or $250$.

Three methods of selecting positive rules were considered:
\begin{enumerate}
\item {\bf Bonferroni:} The full data is used to test the hypotheses
  $H_{0i}:\,r_0 \not \prec_{\Gamma} r_i$ and the Bonferroni correction
  is used to adjust for multiple comparisons.
\item {\bf Ordering by power:} This is the procedure described in
  \Cref{subsec: ranking} using sample splitting and fixed sequence testing.
\item {\bf Ordering by value function:} This is the same as above except that
  the hypotheses are ordered by their estimated value at $\Gamma = 1$.
\end{enumerate}
For the second and third methods, we used either a half or a quarter
of the matched pairs (randomly chosen) to order the
hypotheses. Extra simulation results using different
  split proportions are reported in Supplementary Materials. This simulation was replicated $1000$ times to report the power and the
error rate of the methods. The power is defined as the average size of
the estimated set of positive rules
$\hat{\mathcal{R}}_{\text{pos},\Gamma}$ and the error rate is $1 -
P(\hat{\mathcal{R}}_{\text{pos},\Gamma} \subseteq
\mathcal{R}_{\text{pos}, \Gamma})$ with nominal level
$0.05$.

The results of this simulation study are reported in \Cref{tbl: simu
  res mu_1; 5 cohorts,tbl: simu res mu_2; 5 cohorts,tbl: simu res
  mu_1; 10 cohorts,tbl: simu res mu_2; 10 cohorts}. The error rate is
controlled under the nominal level in most cases and is usually quite
conservative. The conservativeness is not surprising because Rosenbaum's
sensitivity analysis is a worst-case analysis. In terms of power, the
five methods being compared performed very similarly assuming no
unmeasured confounding ($\Gamma =
1$). Bonferroni is still competitive at $\Gamma = 1.5$, but ordering
the hypotheses by (the estimated) power, though losing some sample for
testing, can be much more powerful at larger values of $\Gamma$.
For instance, in \Cref{tbl: simu res mu_1; 10 cohorts} when $\Gamma = 3.0$, two power-based methods are more than twice as powerful as the Bonferroni method. We observe that only using a small planning sample ($25\%$) seems to
  work well in the simulations. This is not too surprising given our
  theoretical results. \Cref{thm: asymp kappa} suggest that only the
  first two moments of $D$ and $|D|$ are needed to estimate the sensitivity
  value asymptotically.

\begin{table}[h]
  \centering
  \captionsetup{singlelinecheck=off}
  \caption[]{Power and error rate (separated by/) of three
    methods estimating
    $\mathcal{R}_{\text{pos},\Gamma}$. Power is defined as the size of
    $\hat{\mathcal{R}}_{\text{pos},\Gamma}$ and
    error rate is defined as $1 -
    P(\hat{\mathcal{R}}_{\text{pos},\Gamma} \subseteq
    \mathcal{R}_{\text{pos}, \Gamma})$. Treatment effect in the
    5 cohorts is given by $\mu = (0.5, 0.25, 0.25, 0.15, 0.05)$. The true $\mathcal{R}_{\text{pos},\Gamma}$ is listed below each $\Gamma$ value.
  }

  \label{tbl: simu res mu_1; 5 cohorts}
  \begin{tabular}{cccccccc}
    \toprule
    \multirow{2}{*}{Cohort size} & \multirow{2}{*}{Method} &$\Gamma = 1.0$ &$\Gamma = 1.8$ &$\Gamma = 2.0$
    &$\Gamma = 2.3$&$\Gamma = 3.0$ \\
                                 & & $\{r_1, \dotsc, r_5\}$ & $\{r_1, \dotsc, r_4\}$ & $\{r_1, r_2, r_3\}$
    & $\{r_1,r_2\}$ & $\{r_1\}$ \\
    \midrule
    \multirow{5}{*}{250}
                                 & Bonferroni
                                                           & 5.00 / 0.00 & 2.54 / 0.01 & 1.60 / 0.03 & 0.72 / 0.02 & 0.11 /
                                                                                                                     0.00 \\
                                 &Value (50\%)
                                                           & 5.00 / 0.00 & 0.51 / 0.07 & 0.08 / 0.04 & 0.00 / 0.00 & 0.00 / 0.00 \\
                                 &Power (50\%)
                                                           & 5.00 / 0.00 & 2.30 / 0.07 & 1.46 / 0.07 & 0.74 / 0.04 & 0.20 / 0.00 \\
                                 &Value (25\%)
                                                           & 5.00 / 0.00 & 0.73 / 0.07 & 0.18 / 0.05 & 0.03 / 0.01 & 0.00 / 0.00 \\
                                 &Power (25\%)
                                                           & 5.00 / 0.00 & 2.64 / 0.07 & 1.69 / 0.06 & 0.85 / 0.05 & 0.21 / 0.00 \\
    \midrule
    \multirow{5}{*}{100}
                                 & Bonferroni
                                                           & 4.99 / 0.00 & 1.39 / 0.02 & 0.75 / 0.02 & 0.37 / 0.02 & 0.08 / 0.00 \\
                                 & Value (50\%)
                                                           & 4.80 / 0.00 & 0.49 / 0.07 & 0.16 / 0.04 & 0.04 / 0.02 & 0.00 / 0.00 \\
                                 & Power (50\%)
                                                           & 4.77 / 0.00 & 1.15 / 0.05 & 0.75 / 0.05 & 0.38 / 0.03 & 0.15 / 0.02 \\
                                 & Value (25\%)
                                                           & 4.99 / 0.00 & 0.61 / 0.06 & 0.24 / 0.03 & 0.12 / 0.03 & 0.01 / 0.00 \\
                                 & Power (25\%)
                                                           & 4.99 / 0.00 & 1.33 / 0.05 & 0.80 / 0.05 & 0.52 / 0.05 & 0.12 / 0.00 \\
    \bottomrule
  \end{tabular}
\end{table}

\begin{table}[h]
  \centering
  \captionsetup{singlelinecheck=off}
  \caption[]{Power and error rate (separated by/) of three
    methods estimating
    $\mathcal{R}_{\text{pos},\Gamma}$. Power is defined as the size of
    $\hat{\mathcal{R}}_{\text{pos},\Gamma}$ and
    error rate is defined as $1 -
    P(\hat{\mathcal{R}}_{\text{pos},\Gamma} \subseteq
    \mathcal{R}_{\text{pos}, \Gamma})$. Treatment effect in the
    5 cohorts is given by $\mu = (0.5, 0.2, -1.0, 0.2, 0.5)$. The true $\mathcal{R}_{\text{pos},\Gamma}$ is listed below each $\Gamma$ value.
  }

  \label{tbl: simu res mu_2; 5 cohorts}
  \begin{tabular}{cccccccc}
    \toprule
    \multirow{2}{*}{Cohort size} & \multirow{2}{*}{Method} &$\Gamma = 1.0$ &$\Gamma = 1.5$ &$\Gamma = 2.0$
    &$\Gamma = 2.5$&$\Gamma = 3.5$ \\
                                 & & $\{r_1,r_2,r_4, r_5\}$ & $\{r_1, r_2, r_5\}$ &$\{r_1,r_2\}$
    & $\{r_1\}$ & $\emptyset$ \\
    \midrule
    \multirow{5}{*}{250}
                                 & Bonferroni
                                                           & 3.14 / 0.00 & 2.36 / 0.00 & 0.78 / 0.00 & 0.45 / 0.01 & 0.01 / 0.01 \\
                                 &Value (50\%)
                                                           & 3.21 / 0.00 & 1.78 / 0.00 & 0.06 / 0.02 & 0.00 / 0.00 & 0.00 / 0.00 \\
                                 &Power (50\%)
                                                           & 3.21 / 0.00 & 2.03 / 0.00 & 0.75 / 0.02 & 0.47 / 0.02 & 0.07 / 0.07 \\
                                 &Value (25\%)
                                                           & 3.25 / 0.00 & 2.21 / 0.00 & 0.07 / 0.03 & 0.00 / 0.00 & 0.00 / 0.00 \\
                                 &Power (25\%)
                                                           & 3.25 / 0.00 & 2.35 / 0.00 & 0.83 / 0.02 & 0.54 / 0.02 & 0.07 / 0.07 \\
    \midrule
    \multirow{5}{*}{100}
                                 & Bonferroni
                                                           & 3.02 / 0.00 & 1.19 / 0.00 & 0.37 / 0.00 & 0.21 / 0.01 & 0.02 / 0.02 \\
                                 & Value (50\%)
                                                           & 3.03 / 0.00 & 0.71 / 0.00 & 0.04 / 0.02 & 0.00 / 0.00 & 0.00 / 0.00 \\
                                 & Power (50\%)
                                                           & 3.02 / 0.00 & 0.93 / 0.00 & 0.43 / 0.01 & 0.29 / 0.04 & 0.08 / 0.08 \\
                                 & Value (25\%)
                                                           & 3.10 / 0.00 & 1.12 / 0.00 & 0.04 / 0.02 & 0.00 / 0.00 & 0.00 / 0.00 \\
                                 & Power (25\%)
                                                           & 3.11 / 0.00 & 1.32 / 0.00 & 0.42 / 0.01 & 0.31 / 0.03 & 0.07 / 0.07 \\
    \bottomrule
  \end{tabular}
\end{table}

\begin{table}[h]
  \centering
  \captionsetup{singlelinecheck=off}
  \caption[]{Power and error rate (separated by/) of three
    methods estimating
    $\mathcal{R}_{\text{pos},\Gamma}$. Power is defined as the size of
    $\hat{\mathcal{R}}_{\text{pos},\Gamma}$ and
    error rate is defined as $1 -
    P(\hat{\mathcal{R}}_{\text{pos},\Gamma} \subseteq
    \mathcal{R}_{\text{pos}, \Gamma})$. Treatment effect in the
    10 cohorts is given by $\mu = (0.5, 0.5, 0.25, 0.25, 0.25, 0.25, 0.15, 0.15, 0.05, 0.05)$. The true $\mathcal{R}_{\text{pos},\Gamma}$ is listed below each $\Gamma$ value.
  }

  \label{tbl: simu res mu_1; 10 cohorts}
  \begin{tabular}{cccccccc}
    \toprule
    \multirow{2}{*}{Cohort size} & \multirow{2}{*}{Method} &$\Gamma = 1.0$ &$\Gamma = 1.8$ &$\Gamma = 2.2$
    &$\Gamma = 3.0$&$\Gamma = 3.5$ \\
                                 & & $\{r_1,\dotsc, r_{10}\}$ & $\{r_1,\dotsc, r_9\}$ &$\{r_1, \dotsc, r_6 \}$
    & $\{r_1, r_2\}$ & $\{r_1\}$ \\
    \midrule
    \multirow{5}{*}{250}
                                 & Bonferroni
                                                           & 10.00 / 0.00 & 6.80 / 0.01 & 2.41 / 0.00 & 0.20 / 0.00 & 0.02 / 0.01 \\
                                 &Value (50\%)
                                                           & 10.00 / 0.00 & 0.88 / 0.06 & 0.00 / 0.00 & 0.00 / 0.00 & 0.00 / 0.00 \\
                                 &Power (50\%)
                                                           & 10.00 / 0.00 & 6.30 / 0.05 & 2.34 / 0.03 & 0.44 / 0.02 & 0.11 / 0.05 \\
                                 &Value (25\%)
                                                           & 10.00 / 0.00 & 1.12 / 0.01 & 0.00 / 0.00 & 0.00 / 0.00 & 0.00 / 0.00 \\
                                 &Power (25\%)
                                                           & 10.00 / 0.00 & 7.06 / 0.06 & 2.73 / 0.02 & 0.42 / 0.02 & 0.10 / 0.05 \\
    \midrule
    \multirow{5}{*}{100}
                                 & Bonferroni
                                                           & 9.99 / 0.00 & 3.97 / 0.01 & 1.14 / 0.00 & 0.12 / 0.00 & 0.03 / 0.02 \\
                                 & Value (50\%)
                                                           & 9.95 / 0.00 & 0.76 / 0.05 & 0.03 / 0.01 & 0.00 / 0.00 & 0.00 / 0.00 \\
                                 & Power (50\%)
                                                           & 9.91 / 0.00 & 3.18 / 0.04 & 1.17 / 0.02 & 0.28 / 0.03 & 0.10 / 0.05 \\
                                 & Value (25\%)
                                                           & 9.95 / 0.00 & 1.06 / 0.04 & 0.06 / 0.01 & 0.00 / 0.00 & 0.00 / 0.00 \\
                                 & Power (25\%)
                                                           & 9.99 / 0.00 & 3.93 / 0.04 & 1.39 / 0.02 & 0.25 / 0.02 & 0.09 / 0.05 \\
    \bottomrule
  \end{tabular}
\end{table}

\begin{table}[h]
  \centering
  \captionsetup{singlelinecheck=off}
  \caption[]{Power and error rate (separated by/) of three
    methods estimating
    $\mathcal{R}_{\text{pos},\Gamma}$. Power is defined as the size of
    $\hat{\mathcal{R}}_{\text{pos},\Gamma}$ and
    error rate is defined as $1 -
    P(\hat{\mathcal{R}}_{\text{pos},\Gamma} \subseteq
    \mathcal{R}_{\text{pos}, \Gamma})$. Treatment effect in the
    10 cohorts is given by $\mu = (0.5, 0.3, 0.2, 0.0, -1.0, -1.0, 0.5, 0.5, 1.0, 1.0)$. The true $\mathcal{R}_{\text{pos},\Gamma}$ is listed below each $\Gamma$ value.
  }

  \label{tbl: simu res mu_2; 10 cohorts}
  \begin{tabular}{cccccccc}
    \toprule
    \multirow{2}{*}{Cohort size} & \multirow{2}{*}{Method} &$\Gamma = 1.0$ &$\Gamma = 1.5$ &$\Gamma = 2.0$
    &$\Gamma = 2.5$&$\Gamma = 3.0$ \\
                                 & & $\{r_1,\dotsc, r_4, r_9, r_{10}\}$ & $\{r_1,\dotsc, r_4\}$ &$\{r_1, r_2, r_3 \}$
    & $\{r_1, r_2\}$ & $\{r_1\}$ \\
    \midrule
    \multirow{5}{*}{250}
                                 & Bonferroni
                                                           & 5.98 / 0.00 & 3.51/0.00 & 1.55 / 0.00 & 0.37 / 0.00 & 0.07 / 0.00 \\
                                 &Value (50\%)
                                                           & 5.97 / 0.02 & 0.10 / 0.02 & 0.00 / 0.00 & 0.00 / 0.00 & 0.00 / 0.00 \\
                                 &Power (50\%)
                                                           & 6.00 / 0.02 & 3.43 / 0.02  & 1.53 / 0.01 & 0.56 / 0.02 & 0.21 / 0.01 \\
                                 &Value (25\%)
                                                           & 6.02 / 0.03 & 0.23 / 0.04 & 0.03 / 0.00 & 0.01 / 0.00 & 0.00 / 0.00 \\
                                 &Power (25\%)
                                                           & 6.02 / 0.02 & 3.67 / 0.04 & 1.84 / 0.01 & 0.66 / 0.01 & 0.22 / 0.01\\
    \midrule
    \multirow{5}{*}{100}
                                 & Bonferroni
                                                           & 5.60 / 0.01 & 2.42 / 0.00 & 0.68 / 0.00 & 0.19 / 0.00 & 0.06 / 0.00\\
                                 & Value (50\%)
                                                           & 5.24 / 0.03 & 0.22 / 0.03 & 0.04 / 0.00 & 0.01 / 0.00 & 0.00 / 0.00\\
                                 & Power (50\%)
                                                           & 5.48 / 0.03 & 2.23 / 0.03 & 0.86 / 0.02 & 0.31 / 0.02 & 0.16 / 0.02 \\
                                 & Value (25\%)
                                                           & 5.58 / 0.02 & 0.71 / 0.03 & 0.16 / 0.00 & 0.04 / 0.00 & 0.01 / 0.00 \\
                                 & Power (25\%)
                                                           & 5.71 / 0.02 & 2.61 / 0.03 & 0.98 / 0.01 & 0.36 / 0.02 & 0.14 / 0.02 \\
    \bottomrule
  \end{tabular}
\end{table}

\section{Application: The effect of late retirement on senior health}
\label{sec:application}

Finally we apply the proposed method to study the potentially
heterogeneous effect of retirement timing on senior health.
Many empirical studies have focused on the effect
of retirement timing on short-term and long-term health status of the
elderly people
\citep{morrow2001productive,alavinia2008unemployment,borsch2006early}. One
theory known as the ``psychosocial-materialist'' approach suggests
that retiring late may have health benefits because work forms a key part of
the identity of the elderly and provides financial, social and
psychological resources \citep{Calvo2012}. However, the health
benefits of late retirement may differ in different
subpopulations \citep{Dave_2008_late_retirement,
  Westerlund_late_retirement_2009}.

We obtained the observational data from the Health and Retirement
Study, an ongoing nationally representative survey of more than 30,000
adults who are older than 50 and their spouses in the United States. HRS
is sponsored by the National Institute of Aging; Detailed information
on the HRS and its design can be found in
\citet{sonnega_IJE_HRS}. We use the RAND HRS Longitudinal File 2014
(V2), an easy-to-use dataset based on the HRS core data
that consists of a follow-up study of $15,843$ elderly people
\citep{RAND_HRS_data}.

We defined the treatment as late retirement (retirement after
$65$ years old and before $70$ years old) and asked how it impacted
self-reported health status
at the age of $70$ (coded by: 5 - extremely good, 4 - very good,
3 - good, 2 - fair, and 1 - poor). We included individuals who retired
before $70$ and had complete measurements of the following confounders: year of
birth, gender, education (years), race (white or not), occupation (1:
executives and managers, 2: professional specialty, 3: sales and
administration, 4: protection services and armed forces, 5: cleaning,
building, food prep, and personal services, 6: production,
construction, and operation), partnered, annual income, and smoking
status. This left us with $1934$
  treated subjects and $4831$ controls. Figure \ref{fig: age_dist}
  plots the distribution of retirement age in all samples and in the
  treatment group. The distribution of retirement age in the treatment
  group is right skewed, with a spike of people retiring shortly after
  $65$ years old. In the Supplementary Materials,
  we give a detailed account of data
  preprocessing and sample inclusion criteria.

\begin{figure}[h]%
  \centering
  \subfigure[In all samples]{%
    \label{fig:first}%
    \includegraphics[height=2.5in]{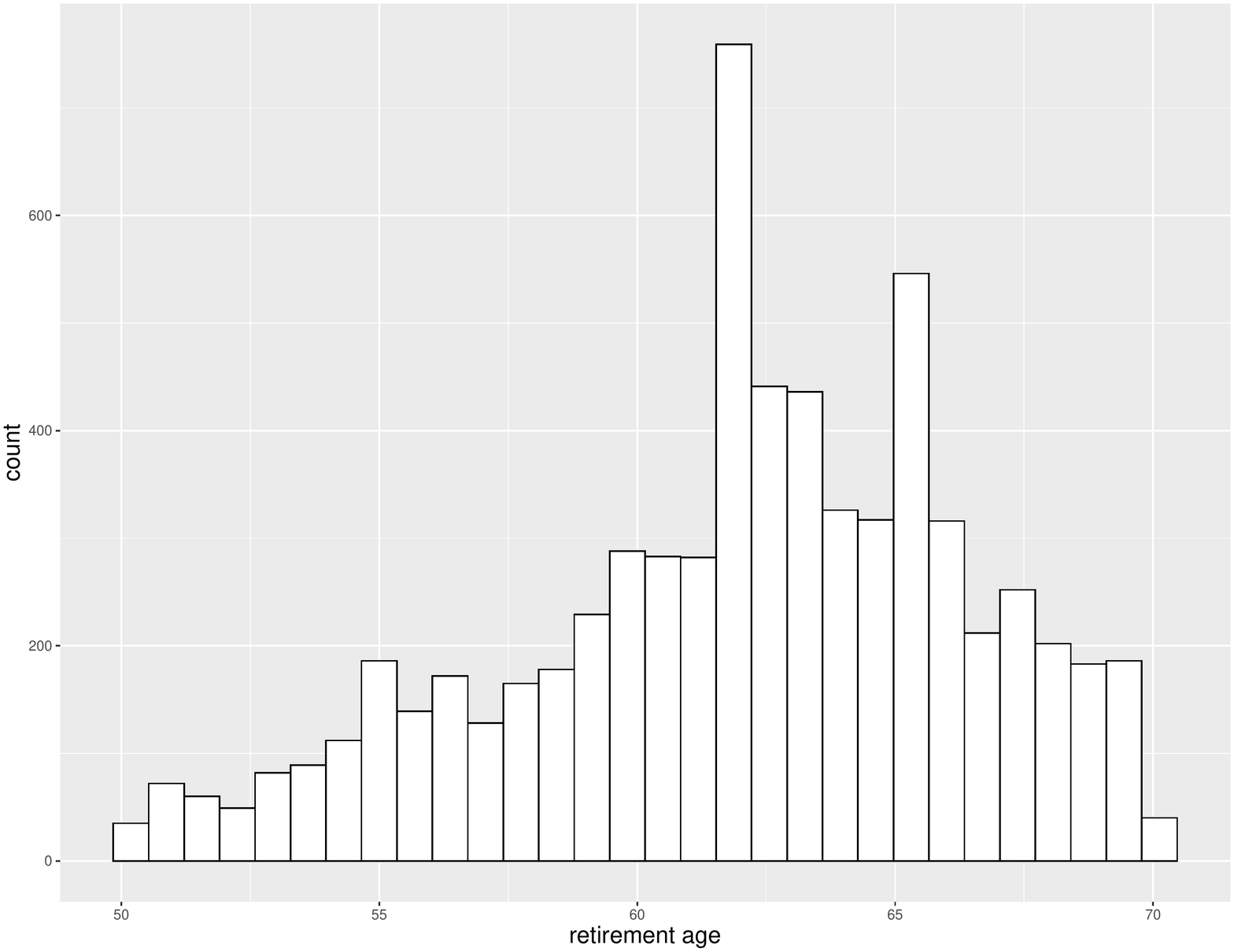}}%
  \subfigure[In the treated group]{%
    \label{fig:second}%
    \includegraphics[height=2.5in]{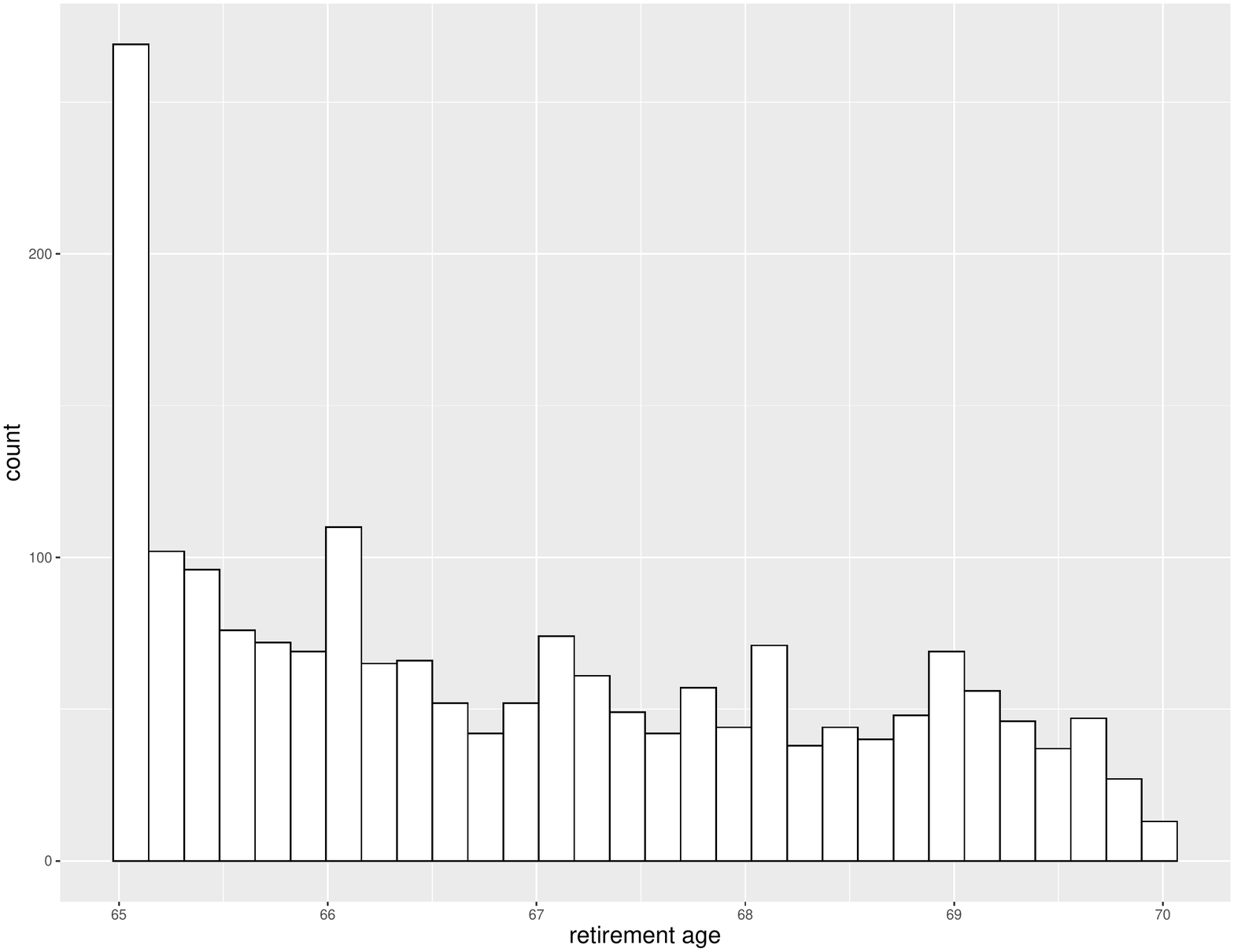}}%
  \caption{Distribution of retirement age}
  \label{fig: age_dist}
\end{figure}

Using optimal matching as implemented in the \texttt{optmatch} R
  package \citep{optmatch}, we form $1858$ matched pairs, matching
  exactly on the year of
  birth, gender, occupation, and partnered or not, and balance the race,
  years of education, and smoking status. Table \ref{tbl: balance table
    retirement} summarizes the covariate balance after
  matching. After matching, the treated and control
  groups are well-balanced (\Cref{tbl: balance table retirement}): the
  standardized differences of all covariates are less than
  0.1. Additionally, the propensity score in the treated and control group
  have good overlap before and after matching (see the Supplementary
  Materials).

\begin{table}[t]
  \centering
  \caption{Covariate balance after matching.}
  \label{tbl: balance table retirement}
  \begin{tabular}{lrrrr}
    \hline
    \hline
    & Control & Treated & std.diff &  \\
    \hline
    \hline
    Year of birth & 1936.27 & 1936.27 & 0.00 &     \\
    Female & 0.53 & 0.53 & 0.00 &     \\
    Non-hispanic white & 0.77 & 0.75 & -0.04 &     \\
    Education (yrs) & 12.52 & 12.53 & 0.00 &     \\
    Occupation: cleaning, building, food prep, and personal services & 0.10 & 0.10 & 0.00 &     \\
    Occupation: executives and managers & 0.16 & 0.16 & 0.00 &     \\
    Occupation: production construction and operation occupations & 0.28 & 0.28 & 0.00 &     \\
    Occupation: professional specialty & 0.19 & 0.19 & 0.00 &     \\
    Occupation: protection services and armed forces & 0.02 & 0.02 & 0.00 &     \\
    Occupation: sales and admin & 0.25 & 0.25 & 0.00 &     \\
    Partnered & 0.74 & 0.74 & 0.00 &     \\
    Smoke ever  & 0.63 & 0.59 & -0.08 &    \\
    \hline
  \end{tabular}
\end{table}

We considered two potential effect modifiers, namely gender and
occupation. More complicated treatment rules can in principle be
considered within our framework, though having more treatment rules
generally reduces the power of multiple testing. We grouped the $6$
occupations into $2$ broad categories: white collar jobs (executives and managers and professional
specialties) and blue collar jobs (sales, administration, protection
services, personal services, production, construction, and
operation). There were $4$ subgroups defined by these two potential
effect modifiers: male, white-collar workers ($G_1$), female,
white-collar workers ($G_2$), male, blue-collar workers ($G_3$), and
female, blue-collar workers ($G_4$). Thus, there were a total of $2^4 =
16$ different regimes formed out of these two effect modifiers. We
gave decimal as well as binary codings to the $16$ groups: $r_0$
($r_{0000}$) assigns control to everyone, $r_8\,(r_{1000}), r_4\,
(r_{0100}), r_2\,(r_{0010}), r_1\,(r_{0001})$
assign treatment to one of the $4$ subgroups, and so forth. 
We split the matched samples and used $1/4$ of them to plan the
test in the other $3/4$. Then we followed the procedures proposed in
\Cref{sec:multiple} to rank and select the treatment rules.

  \Cref{fig: hrs hasse 1.2} reports the estimated Hasse diagram at
  $\Gamma = 1.2$; additional results can be found in the Supplementary
  Materials. The estimated maximal rules for various choices of $\Gamma$ and $\delta$ are reported in
\Cref{tbl: HRS max set result} and the estimated positive rules are
reported in the Supplementary Materials. According to
  \Cref{tbl: HRS max set result}, the maximal rules under the no unmeasured confounding assumption are $r_{11} \, (r_{1011})$ which assigns late retirement
  to all but female, white-collar workers, $r_{13} \, (r_{1101})$ which
  assigns late retirement to all but male, blue-collar workers, and
  $r_{15} \, (r_{1111})$ which assigns treatment to everyone. When
  $\Gamma$ increases to $1.2$, $r_9 \, (r_{1001})$ which assigns
  treatment to male, white-collar workers and female, blue-collar
  workers, further enters the set of maximal rules. The estimated
positive rules suggest that $r_9\,(r_{1001})$ and $r_1 \, (r_{0001})$ which only assigns late retirement to
female blue collar workers, though not among the maximal rules at
$\Gamma = 1$ in \Cref{tbl: HRS max set result}, are the most robust to
unmeasured confounding. This suggests that later retirement perhaps
benefit female blue-collar workers more than others.

\begin{figure}
  \centering
  \includegraphics[scale = 0.5]{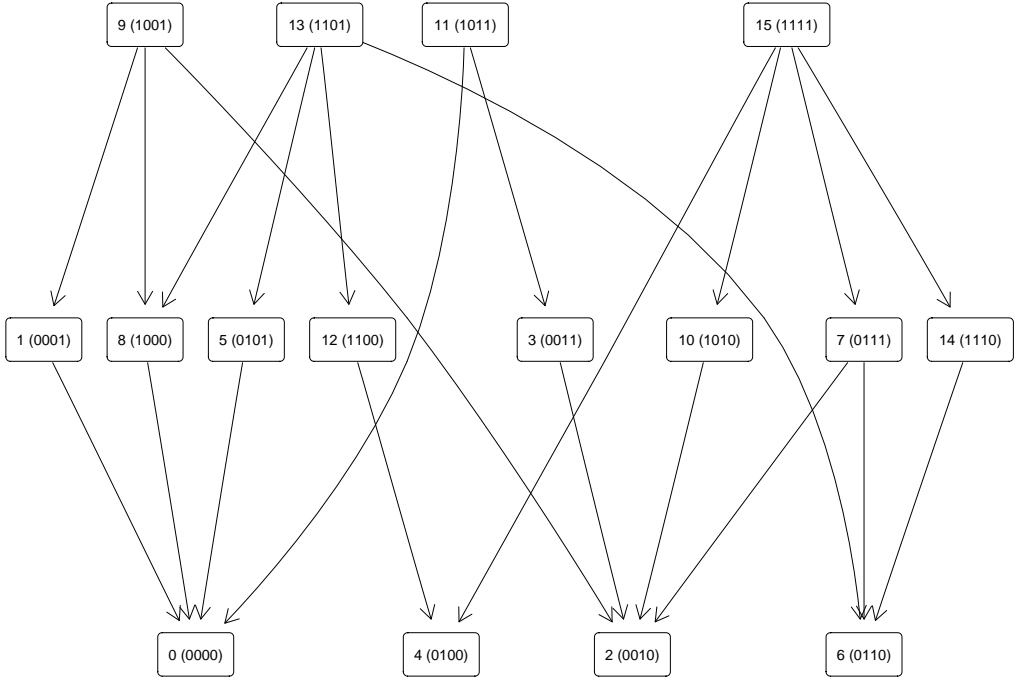}
  \caption{The effect of late retirement on health: Hasse diagram at $\Gamma = 1.2$}
  \label{fig: hrs hasse 1.2}
\end{figure}

\begin{table}[h]
  \centering
  \caption{The effect of late retirement on health: $\hat{\mathcal{R}}_{\text{max}, \Gamma}$ for different choices of $\Gamma$.}
  \label{tbl: HRS max set result}
  \begin{tabular}{cc}
    \toprule
    $\Gamma$ &$\hat{\mathcal{R}}_{\text{max}, \Gamma}$ \\
    \midrule
    1.0     & $\{r_{11}, r_{13}, r_{15}\}$        \\
    1.2     & $\{r_9, r_{11}, r_{13}, r_{15}\}$      \\
    1.35    &  $\{r_1, r_3, r_5, r_7, r_9, r_{11}, r_{13}, r_{15}\}$ \\
    \bottomrule
  \end{tabular}
\end{table}

Does $\Gamma = 1.2$ represent a weak or strong
  unmeasured confounder? \citet{Rosenbaum2009} proposed to
  \emph{amplify} $\Gamma$ to a two-dimensional curve indexed by
  $(\Lambda, \Delta)$, where $\Lambda$ describes the relationship
  between the unmeasured confounder and the treatment assignment, and
  $\Delta$ describes the relationship between the unmeasured confounder and the
  outcome. For instance, $\Gamma = 1.2$ corresponds to an unmeasured
  confounder associated with a doubling of the odds of late retirement
  and a $75\%$ increase in the odds of better health status at the age
  of $70$ in each matched pair, i.e., $(\Delta, \Lambda) = (2.0,
  1.75)$. \citet{Hsu2013} further proposed to calibrate $(\Lambda,
  \Delta)$ values to coefficients of observed
  covariates, however, their method only works for binary outcome and
  binary treatment. In the Supplementary Materials, we describe a
  calibration analysis that handle the ordinal self-reported health
  status level in our application that has $5$ levels.

  We follow \citet{Hsu2013} and use a plot to
  summarize the calibration analysis. In Figure \ref{fig: calibration
    plotb}, the
  blue curve represents \citet{Rosenbaum2009}'s two-dimensional amplification of
  $\Gamma = 1.2$ indexed by $(\Lambda, \Delta)$. The estimated
  coefficients of observed covariates are represented by black dots
  (after taking an exponential so they are comparable to $(\Lambda,
  \Delta)$). We followed the suggestion in \cite{Gelman2008} and
  standardized all the non-binary covariates to have mean $0$ and
  standard deviation $0.5$, so the coefficient of each binary variable
  can be interpreted directly and the coefficients of each
  continuous/ordinal variable can be interpreted as the effect of a
  2-SD increase in the covariate  value, which roughly corresponds to
  flipping a binary variable from $0$ to $1$. Note that all
  coefficients are under the $\Gamma = 1.2$ curve. In fact,
  $\Gamma = 1.2$ roughly corresponds to a moderately strong binary
  unobserved covariate whose effects on late retirement and health
  status are comparable to a binary covariate $U$ constructed from
  smoking and education (red star in \Cref{fig: calibration plotb}).
\begin{figure}
  \centering
  \includegraphics[scale = 0.5]{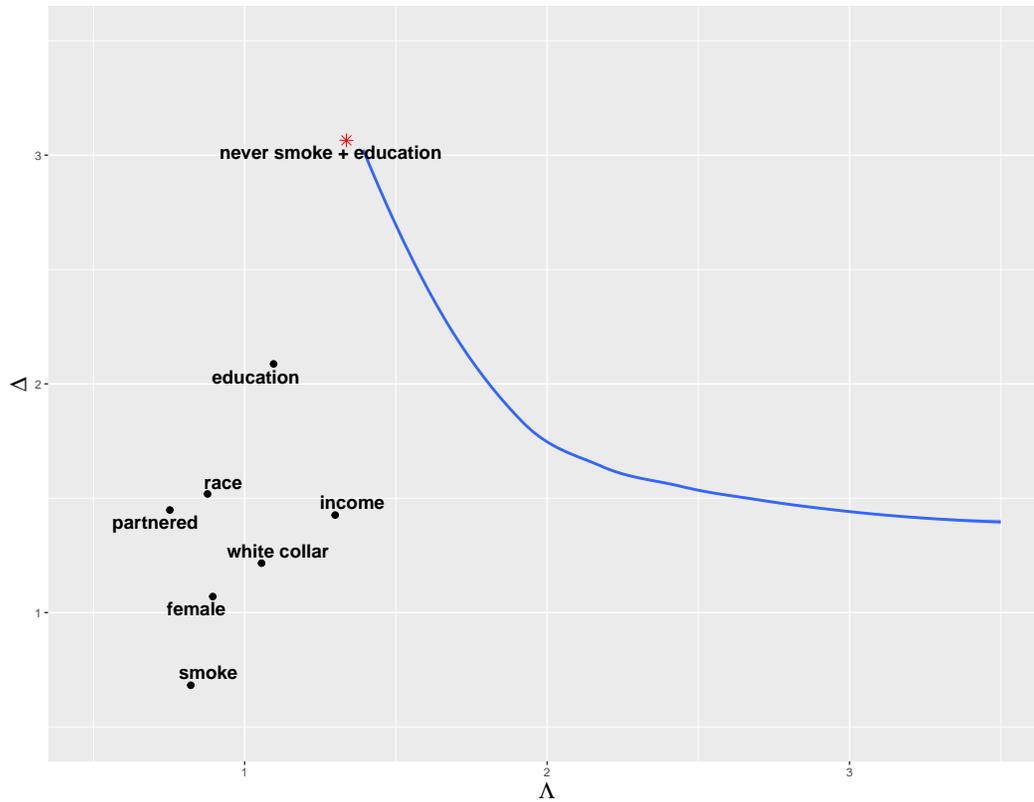}
  \caption{The effect of latent retirement on health: Calibration of the sensitivity
      analysis. The blue curve is \citet{Rosenbaum2009}'s
      amplification of $\Gamma = 1.2$. Black dots represent estimated
      coefficient of each observed covariate. Red marker represents
      the aggregate effect of flipping smoking status and increasing
      education by 2 standard deviations.}
  \label{fig: calibration plotb}
\end{figure}

\section{Discussion}
\label{sec:discussion}

In this paper we proposed a general framework to compare, select, and
rank treatment rules when there is a limited degree of unmeasured
confounding and illustrated the proposed methods by two real data
examples. A central message is that the best treatment rule (with the
largest estimated value) assuming no unmeasured confounding is often
not the most robust to unmeasured confounding. This may have important
policy implications when individualized treatment rules are learned
from observational data.

Because the value function only defines a partial order on the
treatment rules when there is unmeasured confounding, there is a
multitude of statistical questions one can ask about selecting
and ranking the treatment rules. We have considered three questions
that we believe are most relevant to policy research, but there are
many other questions (such as in \citet{gibbons1999selecting}) one
could ask.

In principle, our framework can be used with
  an arbitrary number of prespecified individualized treatment
  rules. However, to maintain a good statistical power in the
  multiple testing, the prespecified treatment rules should not be too
  many. This limitation makes our method most suitable as a
  confirmatory analysis to complement machine learning algorithms
  for individualized treatment rule discovery. Alternatively, if
  the number of decision variables is relatively low due to economic
  or practical reasons, our method is also reasonably powered
  for treatment rule discovery.

\section*{Acknowledgement}

JW received funding from the Population Research Training Grant (NIH
T32 HD007242) awarded to the Population Studies Center at the
University of Pennsylvania by the NIH's Eunice Kennedy Shriver
National Institute of Child Health and Human Development.

\appendix
\section{Appendix: Proof of \Cref{thm: asymp kappa}}

To simplify the notation, suppose $r_1(\bm X_i) < r_2(\bm X_i)$ for
all $1 \le i \le I$. Let \[
  D_{i, \Gamma} = D_i - \left(\frac{\Gamma - 1}{\Gamma + 1}\right)|D_{i, \Gamma}|, \quad \overline{D} = \frac{1}{I}\sum_{i = 1}^{I} D_i, \quad \overline{|D|} = \frac{1}{I}\sum_{i = 1}^{I} |D_i|,
\]
\[
  \overline{D}_\Gamma = (1/I)\sum_{i = 1}^{I} D_{i, \Gamma} = \overline{D} - \left(\frac{\Gamma - 1}{\Gamma + 1}\right) \overline{|D|},
\]
and \[
  se(\overline{D}_\Gamma)^2 = \frac{1}{I^2} \sum_{i = 1}^{I} (D_{i, \Gamma} - \overline{D}_\Gamma)^2.
\]

When $\mathbb{E}[D_i] > 0$, $\Gamma^\ast(r_1, r_2)$ is obtained by solving the equation below in $\Gamma$:
\begin{equation}
  \frac{\overline{D}_\Gamma}{se(\overline{D}_\Gamma)} = \Phi^{-1}(1 - \alpha).
  \label{eqn: sens value}
\end{equation}
Square both sides of the equation above and plug in the expressions for $\overline{D}_\Gamma$ and $se(\overline{D}_\Gamma)^2$. Let $\kappa = (\Gamma - 1)/(\Gamma + 1)$ and $z_\alpha = \Phi^{-1}(1 - \alpha)$. Denote \[
  z_\alpha = \Phi^{-1}(1-\alpha), \quad \overline{|D|} = \frac{1}{I}\sum_{i = 1}^{I } |D_i|,\quad ,
\] and
\[
  s^2_{D} = \frac{1}{I}\sum_{i = 1}^{I} (D_i - \overline{D})^2,\quad s^2_{|D|} = \frac{1}{I}\sum_{i = 1}^{I} (|D_i| - \overline{|D|})^2,\quad s_{D, |D|} = \frac{1}{I}\sum_{i = 1}^{I} (D_i - \overline{D})(|D_i| - \overline{|D|}).
\]

One can show the sensitivity value $\Gamma^\ast(r_1, r_2)$ corresponds to $\kappa^\ast$ that solves the following quadratic equation:
\[
  \left(\overline{|D|}^2 - \frac{1}{I} s^2_{|D|}z_\alpha^2 \right) \kappa^2 - 2\left(\overline{D}\overline{|D|} - \frac{1}{I} s_{D, |D|} z_\alpha^2\right)\kappa + \overline{D}^2 - \frac{1}{I} s^2_{D}c^2_\alpha = 0.
\]

Specifically, we have
\begin{equation}
  \kappa^\ast = \frac{{\overline{D}\overline{|D|} - \frac{1}{I} s_{D, |D|} z_\alpha^2} \pm \sqrt{\Delta} }{\overline{|D|}^2 - \frac{1}{I} s^2_{|D|}z_\alpha^2},
  \label{eqn: quadratic root}
\end{equation}
where $\Delta = (\overline{D}\overline{|D|} - \frac{1}{I} s_{D, |D|} z_\alpha^2)^2 - (\overline{|D|}^2 - \frac{1}{I} s^2_{|D|}z_\alpha^2)(\overline{D}^2 - \frac{1}{I} s^2_{D}c^2_\alpha)$.

Note
\[
  \sqrt{\Delta} = z_\alpha \sqrt{\frac{1}{I} \left(s^2_{|D|}\cdot \overline{D}^2 + s^2_{D}\cdot \overline{|D|}^2 -2 \overline{D}\overline{|D|} s_{D, |D|}\right) + \frac{1}{I^2}z_\alpha^2 \left(s_{D, |D|}^2 - s^2_{|D|}\cdot s^2_{D}\right)}.
\]
Let us denote $A = \mathbb{E}[D]\cdot\mathbb{E}[|D|]$, $B = -z_\alpha \sqrt{\sigma^2_{|D|}\cdot \mathbb{E}[D]^2 + \sigma^2_{D}\cdot \mathbb{E}[|D|]^2 -2 \mathbb{E}[D]\mathbb{E}[|D|] \sigma_{D, |D|}}$, $C = \mathbb{E}[|D|]^2$, $R_1 = \sqrt{I}(\overline{D}\overline{|D|} - A)$, $R_2 = \sqrt{I}(\overline{|D|}^2 - C)$.

We have \[
  \kappa^\ast = \frac{A + \frac{1}{\sqrt{I}}R_1 + \frac{1}{\sqrt{I}} B}{C + \frac{1}{\sqrt{I}} R_2} + o_p\left(\frac{1}{\sqrt{I}}\right) = \frac{(A + \frac{1}{\sqrt{I}}R_1 + \frac{1}{\sqrt{I}} B)\cdot(1 - \frac{1}{\sqrt{I}}\frac{R_2}{C})}{C} + o_p\left(\frac{1}{\sqrt{I}}\right).
\]

Scale both sides by $\sqrt{I}$ and rearrange the terms, we have\[
  \sqrt{I}\left(\kappa^\ast - \frac{A}{C}\right) = \frac{B}{C} + \frac{1}{C}R_1 - \frac{A}{C^2} R_2 + o_p(1).
\]

Moreover, let $\phi: \mathbb{R}^2 \mapsto \mathbb{R}^2 = (xy, y^2)$:
\[
  \sqrt{I}\begin{pmatrix}
    \overline{D} - \mathbb{E}[D] \\
    \overline{|D|} - \mathbb{E}[|D|]
  \end{pmatrix} \sim N(0, \Sigma),
  \quad \text{implies} \quad
  \begin{pmatrix}
    R_1 \\
    R_2
  \end{pmatrix} = \sqrt{I}\begin{pmatrix}
    \overline{D}\overline{|D|} - \mathbb{E}[D]\cdot\mathbb{E}[|D|] \\
    \overline{|D|}^2 - \mathbb{E}[|D|]^2
  \end{pmatrix}\sim N(0, \phi'\Sigma (\phi')^T),
\] where $\Sigma = \begin{pmatrix}
  \text{Var}[D], & \text{Cov}(D, |D|)\\
  \text{Cov}(D, |D|), & \text{Var}[|D|]
\end{pmatrix}$ and $\phi' = \begin{pmatrix}
  \mathbb{E}[|D|], &\mathbb{E}[D] \\
  0, &2\mathbb{E}[|D|]
\end{pmatrix}$.

Plug in the expressions for $A$, $B$, and $C$ and compute the variance-covariance matrix of $(1/C, -A/C^2)(R_1, R_2)^T$: \[
  \sqrt{I}\left(\kappa^\ast - \frac{\mathbb{E}[D]}{\mathbb{E}[|D|]}\right) \sim N(z_\alpha\mu, ~\sigma^2)
\]
where
\[
  \mu = -\frac{\sqrt{\sigma^2_{|D|}\cdot \mathbb{E}[D]^2 + \sigma^2_{D}\cdot \mathbb{E}[|D|]^2 -2 \mathbb{E}[D]\mathbb{E}[|D|] \sigma_{D, |D|}}}{\mathbb{E}[|D|]^2},
\]
and
\[
  \sigma^2 = \frac{\text{Var}[D] \mathbb{E}^2[|D|] - \text{Var}[|D|] \mathbb{E}^2[D] - 2\mathbb{E}[D]\mathbb{E}[|D|]\text{Cov}(D, |D|) + 2\mathbb{E}^2[D]\text{Var}[|D|]}{\mathbb{E}^4[|D|]}.
\]

The Supplementary Materials contain additional appendices about matching in
observational studies and further simulation and real data results.

\clearpage
\bibliographystyle{plainnat}
\bibliography{precision1}

\end{document}